\documentclass[pra,floatfix,amsmath,groupedaddress,twocolumn]{revtex4}
\usepackage{amssymb}
\usepackage{graphicx}
\usepackage{graphics}
\usepackage{amsmath}
\usepackage{mathtools}
\usepackage{amsthm}
\usepackage{amsfonts}
\usepackage{url}
\usepackage{booktabs}
\usepackage{braket}
\usepackage{bm}
\usepackage{bbold}
\usepackage{float}
\usepackage{enumerate}
\usepackage{color}
\usepackage{hyperref}
\usepackage{multirow}
\usepackage{array}

\newcommand{\bea}{\begin{eqnarray}}
\newcommand{\eea}{\end{eqnarray}}

\def\bi{\begin{itemize}}
\def\ei{\end{itemize}}
\def\bc{\begin{center}}
\def\ec{\end{center}}

\def\C{\hbox{$\mit I$\kern-.7em$\mit C$}}
\def\R{\hbox{$\mit I$\kern-.6em$\mit R$}}

\def\ket#1{|#1\rangle}

\def\tr{\mathrm{tr}}
\def\ket#1{\left| #1\right>}
\def\bra#1{\left< #1\right|}

\newcommand{\proj}[1]{\ket{#1}\bra{#1}}

\newtheorem{theorem}{Theorem}

\newtheorem{lemma}{Lemma}

\newtheorem{definition}{Definition}

\newtheorem{observation}{Observation}

\begin{document}

\author{K. Schwaiger}
\author{B. Kraus}
\affiliation{Institute for Theoretical Physics, University of Innsbruck, Innsbruck, Austria}
\title{Relations between bipartite entanglement measures }

\begin{abstract}
We investigate the properties and relations of two classes of operational bipartite and multipartite entanglement measures, the so-called source and the accessible entanglement. The former measures how easy it is to generate a given state via local operations and classical communication (LOCC) from some other state, whereas the latter measures the potentiality of a state to be convertible to other states via LOCC. Main emphasis is put on the bipartite pure states, single copy regime. We  investigate which parameter regime is physically available, i.e for which values of these measures does there exist a bipartite pure state. Moreover, we determine, given some state, which parameter regime can be accessed by it and from which parameter regime it can be accessed. We show that this regime can be determined analytically using the Postitivstellensatz. Moreover, we compute the boundaries of these sets and the boundaries of the corresponding source and accessible sets. Furthermore, we relate these results to other entanglement measures and compare their behaviors. Apart from that, an operational characterization of bipartite pure state entanglement is presented.
\end{abstract}
\maketitle

\section{Introduction}

Biparite entanglement, in particular pure state entanglement, is considered to be very well understood. This assessment steams mainly from the facts that (i) in the asymptotic regime the entanglement can be completely characterized via the entanglement of formation  \cite{Be} and (ii)  in the single copy case a complete set of entanglement measures for pure states is known \cite{vidal2}. To be more precise, the rate with which $n$ copies of a pure state, $\ket{\Psi} \in \mathbb{C}^d\otimes \mathbb{C}^d$ can be transformed asymptotically and reversibly into the maximally entangled state, $\ket{\Phi^+}=\sum_i \ket{ii}$ is given by the entanglement of formation \cite{Be}. In the single copy regime it is known that the entanglement monotones presented in \cite{vidal2} completely characterize the entanglement contained in the state. That is given the $d-1$ entanglement monotones a unique state (up to local unitary operations (LUs)) is characterized. However, in contrast to the entanglement of formation, the entanglement monotones do not have a clear physical meaning. Hence, despite the fact, that the characterization of bipartite pure state entanglement is mathematically well understood, it is lacking a clear operational meaning. 

However, such a clear physical picture as in the asymptotic case is also highly desirable for the single copy case. This is not only for the sake of understanding bipartite entanglement better from an operational point of view, but, probably more importantly, also to understand certain aspects of correlations in multipartite systems. A prominent example, where the relevance of bipartite entanglement within multipartite systems is shown, is the fundamental area law proven within condensed matter physics \cite{Hastings}. It gives a bound on the bipartite entanglement for ground states of local Hamiltonians (in 1D). This result has been used to prove an efficient description of the ground state of the system in terms of tensor network states, leading to efficient numerical simulations of such systems, which disclosed new insights about the underlying physics \cite{condensed}. 

Hence, a better understanding of bipartite entanglement from an operational point of view might also be relevant is these fields of research. It is precisely the aim of this paper to shine new light onto the physical characterization of bipartite entanglement.

Here, we first show that there is a very simple way to operationally characterize the entanglement contained in a bipartite pure $2^k \times 2^k$ system. It consists of considering such a state as a $2 k$ qubit state and determining the entanglement (measured by only one parameter, e.g. the geometric measure of entanglement \cite{Shim95}) in all possible biparite splittings. We show that these measures determine the Schmidt coefficients of the state uniquely. As any bipartite state can be embedded in a  $2^k \times 2^k$ system, this interpretation can be applied to any bipartite pure system. Moreover, such a characterizaton is also best suited for the case, where, as explained above, a bipartite splitting of a multiqubit state is considered. After that we pursue a different approach, which is based on entanglement measures which have a very clear operational meaning. We recently introduced two classes of entanglement measures, which are applicable to any system size as well as to pure and mixed states, the source entanglement, $E_s$ and the accessible entanglement, $E_a$. We determined these measures in case of multiparite systems in \cite{Schw15} and in case of bipartite system in \cite{Sauer15}. In order to explain the physical meaning let us denote by $M_s (\ket{\Psi})$ the set of states which can be transformed into the state $\ket{\Psi}$ and by $M_a (\ket{\Psi})$ the set of states into which $\ket{\Psi}$ can be transformed to via LOCC. Using then an arbitrary measure the source entanglement, $E_s(\ket{\Psi})$, measures the volume of the set $M_s (\ket{\Psi})$. That is, it measures how many states can be transformed into $\ket{\Psi}$. The more states are in $M_s (\ket{\Psi})$ the less entangled the state is, as all states in $M_s (\ket{\Psi})$ are at least as entangled as $\ket{\Psi}$. The accessible entanglement is then given by the volume of $M_a (\ket{\Psi})$.  In \cite{Sauer15} we derived closed expressions for the source entanglement and showed how the accessible entanglement can be evaluated. \par

It is the aim of this paper to better understand which values these measures can take and how they are related to each other and to other entanglement measures, such as the entanglement of formation. We will present several numerical results concerning the allowed region these measures can take and will then show how these regions can be determined analytically. To this end we use the Positivstellensatz, which characterizes those sets of polynomial equations and inequalities which have a real solution. Moreover, we will explain how the mapping of this problem to a SDP presented in \cite{Parr03} can be achieved in our case. The analysis performed here might be also used in order to obtain a bound on e.g. the value of the source entanglement given the geometric measure of entanglement. The results presented here show that whenever the considered functions are complete, in the sense that they uniquely characterize a state of interest, then the boundaries of these regions can be easily determined. However, in contrast to previous investigations, where different entanglement measures have been considered \cite{BeSa03}, the boundaries are more involved otherwise. Moreover, we investigate the entanglement contained in the states which belong to the source and the accessible set of some given state. There, again, the boundaries of these sets are easily characterized in case the measures are unique (as we will see in the case of $3\times 3$ and $4\times 4 $ systems). Finally we will also consider probabilistic transformations in this context.

The outline of the remainder paper is the following. We first recall the definition of the classes of entanglement measures we consider. 
Moreover, we recall an important theorem in real analysis, the Positivstellensatz, which gives necessary and sufficient conditions for the existence of a real solution to a set of polynomial equations and inequalities in $n$ variables. We also explain how a relaxation of the problem can be solved efficiently for a fixed degree using semidefinite programms (SDP). Furthermore, we present an operational characterization of bipartite pure state entanglement. Next, we analyze the mathematical properties of the entanglement measures and show e.g. that despite the fact that these are entanglement measures, they are not entanglement monotones. The only other example of such measures are, up to the knowledge of the authors, the Renyi entropies for $\alpha >1$, $\alpha \neq \infty$. After that, we focus in Sec. \ref{Sec:RelEsEa} and Sec. \ref{Sec:RelEsEaotherE} on the possible values of $E_s$ and $E_a$ which can be reached from a physical state. We compare these results to previous investigations \cite{BeSa03}, where mainly entropic functions are considered as measures and show the differences to the ones investigated here. The values of the entanglement measures which can be taken by states in the source and accessible set of a given state show a very interesting behavior. We then study how the entanglement can be transformed on the course of a LOCC protocol. After the presentation of these numerical investigation we show how these sets can be obtained analytically using the Positivstellensatz in Sec. \ref{physregionEsEa}. Moreover, we show how this theorem can be used to prove that a set of entanglement measures is complete. 
Finally, we also consider probabilistic transformations and study the entanglement in this setting.

\section{Notations and Preliminaries}
\label{Sec:NotPre}
In this section we first introduce our notations and basic properties of bipartite states and restate the source and accessible entanglement of bipartite states introduced in \cite{Sauer15}. Then we review a theorem from real algebra, namely the \textit{Positivstellensatz}, that gives a necessary and sufficient condition for the existence of a solution of a set of polynomial equations and inequalities and discuss its relevance in the context of this work. \par
Every pure state of a bipartite quantum system with Hilbert space $\mathcal{H} = \mathbb{C}^{d_1} \otimes \mathbb{C}^{d_2}$
can be (up to LUs) written as $\ket{\psi} = \sum_{i=1}^d \sqrt{\lambda_i} \ket{ii}$, i.e. $\ket{\psi} \simeq_{LU} \sum_{i=1}^d \sqrt{\lambda_i} \ket{ii}$,
where $d = \min\{d_1,d_2\}$ and $\lambda_i \geq 0$ denote the Schmidt coefficients with $\sum_i \lambda_i = 1$. We denote by
$\lambda(\psi) = (\lambda_1, \ldots, \lambda_d) \in \R^d$
the Schmidt vector of $\ket{\psi}$ and will consider in the following w.l.o.g. two $d$-level systems, which we call from now on $d \times d$ states. Note that we will often refer to a state via its Schmidt vectors. As can be easily seen from the Schmidt decomposition given above, two $d \times d$ states, $\ket{\psi},\ket{\phi},$ are LU equivalent if and only if (iff)
$\lambda^\downarrow (\psi)=\lambda^\downarrow (\phi)$, where here and in the following $\lambda^\downarrow (\psi)\in \R^d$, with
 $\lambda^\downarrow(\psi)_{i} \geq \lambda^\downarrow(\psi)_{i+1} \geq 0$ denotes the sorted Schmidt vector of $\ket{\psi}$. Note that when considering LOCC transformations of bipartite states, we exclude LU-transformations, as they do not alter the entanglement of the states. That is we pick one representative of each LU-equivalence class, i.e. $\lambda^\downarrow (\psi)$, and consider transformations among these representatives. In the context of
 LOCC transformations of pure bipartite states, the following functions of $x = (x_1,\ldots,x_d) \in \R^d$,
 \begin{align}
  E_k(x) := \sum_{i=k}^d x_i, \ k \in \{1,\ldots,d\},
 \end{align}
 play an important role. It was shown in \cite{nielsen} that a state $\ket{\psi} \in \mathcal{H}$ can be transformed into $\ket{\phi} \in \mathcal{H}$ deterministically via LOCC iff $\lambda^\downarrow(\psi)$ is majorized by $\lambda^\downarrow(\phi)$, written $\lambda^\downarrow(\psi) \prec \lambda^\downarrow(\phi)$, i. e.
 \begin{align} \label{Ineq_Maj} E_k(\lambda^\downarrow(\psi)) \geq E_k(\lambda^\downarrow(\phi)) \ \forall k \in \{1,\ldots,d\}, \end{align}
 with equality for $k=1$.
A direct consequence of this criterion is that the source and accessible set, i.e. the set of states that can either reach or be reached by a state $\ket{\psi}$ deterministically via LOCC, of $\ket{\psi} \in \mathcal{H}$ are given by
\begin{align}
 M_s(\psi) = \{\ket{\phi}\in {\cal H} \ \mbox{s.t.} \ \lambda (\phi) \prec \lambda (\psi)\}, \nonumber\\
 M_a(\psi) = \{\ket{\phi}\in {\cal H} \ \mbox{s.t.} \ \lambda (\psi) \prec \lambda (\phi)\}. \label{MaMs}
\end{align}
Let us now review the idea of measuring the volume of the source and accessible set, i.e. $M_s (\psi)$ and $M_a (\psi)$, and thus, obtaining a valid operational entanglement measure, as mentioned in the introduction. Let $\mu$ denote an arbitrary measure in the set of LU-equivalence classes. As mentioned above we consider one representative of each LU-equivalence class in $M_a$ and $M_s$, as we do not consider LU-transformations.  Then, the source volume is defined by $V_s(\psi)=\mu[M_s(\psi)]$ and the accessible volume by $V_a(\psi)=\mu[M_a(\psi)]$. Hence, the accessible and source entanglement are given by 
\begin{equation}\label{eas}
E_{a}(\psi)=\frac{V_a(\psi)}{V_a^{sup}},\quad E_{s}(\psi)=1-\frac{V_s(\psi)}{V_s^{sup}},
\end{equation}
where $V_a^{sup}$ ($V_s^{sup}$) denote the supremum of the accessible (source) volume according to the measure $\mu$. Note that for any valid measure $\mu$, $E_s$ and $E_a$ are valid entanglement measures, i.e. they are not increasing under LOCC. This can be easily verified considering the operational meaning of $E_s$, $E_a$. Whereas $E_s$ measures how easy it is to obtain a certain state via LOCC, $E_a$ measures how useful a state is, as this state is at least as powerful as any state in its accessible set. For bipartite pure states we presented in \cite{Sauer15} the following closed expression for the source entanglement \footnote{In the derivation of the source entanglement and its generalizations for bipartite pure states we picked a certain measure $\mu$ in the set of LU-equivalence classes to obtain the closed expression given in Eq.\ \eqref{Essimp}). As the source (and accessible) set fulfills the necessary LOCC-monotonicity condition we could use any measure $\mu$ to compute the volume of this set and obtain a valid entanglement measure.}
\begin{equation}
E_s(\psi) = 1- \sum_{\sigma \in \Sigma_d} \frac{(\sum_{k=1}^d \sigma_k \lambda_k - \frac{d+1}{2})^{d-1}}{\prod_{k=1}^{d-1} (\sigma_k - \sigma_{k+1})},
\label{sourceent}
\end{equation}
where $\Sigma_d$ denotes the symmetric group and the sum runs over all elements of this group. Note that for the accessible entanglement we have closed expressions for low-dimensional bipartite systems, e.g. for $3 \times 3$ states $E_a(\psi) =   \begin{cases} 12 \lambda_2 \lambda_3 \  \textrm{if } \lambda_1 > \frac{1}{2} \\
12 [\lambda_2 \lambda_3 - 1/4 (1-2\lambda_1)^2]  \ \textrm{if } \lambda_1\leq\frac{1}{2}
  \end{cases}$ (see Appendix \ref{sec:AppendixA}). Furthermore, we provided algorithms to compute the accessible entanglement numerically in \cite{Sauer15} for higher dimensional states.  For $3 \times 3$ states the source and the accessible set in terms of the Schmidt coefficients is shown in Fig.\ \ref{3x3volumes}. We reconsider the boundaries of these sets in Fig.\ \ref{3x3volumes} in the following sections. Note that similar figures were already introduced in \cite{ZyBe02}, where the authors investigated geometric properties of bipartite entanglement in terms of the Schmidt coefficients.
\begin{figure}[h!]
\centering
\includegraphics[width=0.2\textwidth]{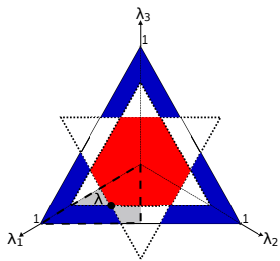}
\caption{\cite{Sauer15} The source and the accessible set of a $3 \times 3$ state with Schmidt vector $\lambda_{\phi} = (0.6, 0.37,0.13)$ in terms of the Schmidt coefficients. The thick, dashed line encloses the set of sorted Schmidt vectors that is in one-to-one correspondence to the LU-equivalence classes. The red (blue) regions depict the source (accessible) set of the quantum state respectively. Note that the number of vertices of the accessible set can change and is in this case equal to four (see \cite{Sauer15} for details). }
\label{3x3volumes}
\end{figure}
Note, further, that in \cite{Sauer15} we introduced also the generalization of the source entanglement, which leads to a whole class of operational entanglement measures. For these generalizations we measure the set of states, $\ket{\Psi} \in \mathbb{C}^k \otimes \mathbb{C}^k$, that can be converted via LOCC to a $d \times d$ state $\ket{\psi}$, with smaller or equal dimensions, i.e. $d \leq k$. Then we can identify the state $\ket{\psi}$ with a state $\ket{\Psi^k(\psi)} \in \mathbb{C}^k \otimes \mathbb{C}^k$, whose Schmidt vector is simply given by adding $k-d$ zeros to the initial, d-dimensional Schmidt vector of $\ket{\psi}$. Thus, we get a whole class of operational entanglement measures given by the generalizations of the source entanglement, that are for each $k$ given by (see \cite{Sauer15} for more details)
\begin{align}
E_s^{k\rightarrow d}(\psi) = \frac{1}{\sup\limits_{\ket{\phi} \in  \mathbb{C}^d \otimes \mathbb{C}^d} E_s(\Psi^k(\phi))} E_s(\Psi^k(\psi)) \quad ,k \geq d. \label{eq:SourceGeneralization}
\end{align}
Note that the dimension of the volumes corresponding to the generalizations of the source entanglement is higher than the dimension of the source volume corresponding to the source entanglement in Eq.\ \eqref{sourceent}. That is the dimension of the source sets corresponding to the generalized source entanglement $E_s^{k\rightarrow d}$ is equal to $k-1$. 
\par
Let us now review the Positivstellensatz \cite{Steng74}, which is a fundamental theorem in real algebraic geometry. It states that there either exists a polynomial identity, which certifies that a system of polynomial equations and inequalities has no solution in $\mathbb{R}^n$ or there exists indeed a solution to this system. A few definitions of algebraic objects, i.e. the ideal and the cone, are in order before we can recall the theorem. 
\begin{definition}
The subset $I \subseteq \mathbb{R}[x_1,...,x_n]$ is an \textup{ideal} if it satisfies
\begin{enumerate}[(a)]
\item 0 $\in I$,
\item If a, b $\in I$, then a+b $\in I$,
\item If a $\in I$ and b $\in \mathbb{R}[x_1,...,x_n]$ , then a$\cdot$ b $\in I$.
\end{enumerate}
\end{definition}
A simple example would be the ideal corresponding to the set of multivariate polynomials $\{h_1,...,h_m\}$, $h_i \in \mathbb{R}[x_1,...,x_n] \forall i$. It is given by
\begin{equation}
I(h_1,...,h_m) = \{ h \mid h= \sum_{i=1}^m t_i h_i, \ t_i \in \mathbb{R}[x_1,...,x_n]\}. 
\end{equation}
Another example of an ideal is the set of polynomials with a common set of roots. 
\begin{definition}
The subset $P \subseteq \mathbb{R}[x_1,...,x_n]$ is a \textup{cone} if it satisfies
\begin{enumerate}[(a)]
\item If a,b $\in P$, then a + b $\in P$,
\item If a, b $\in P$, then a $\cdot$ b $\in P$,
\item If a  $\in \mathbb{R}[x_1,...,x_n]$ , then a$^2$ $\in P$.
\end{enumerate}
\end{definition}
Hence, the cone of a set of multivariate polynomials $\{f_1,...,f_s\}$,  $f_i \in \mathbb{R}[x_1,...,x_n] \forall i$, reads
\begin{equation}
\begin{split}
P(f_1,...,f_s) = \{f \mid f = & s_0 + \sum_{i=1}^m s_i f_i + \sum_{i<j} s_{i j } f_i f_j + \\ &\sum_{i<j<k} s_{i j k} f_i f_j f_k + ... s_{1 2 ... s} f_1 f_2 ... f_s\},
\end{split}
\end{equation}
with $s_{\{i_1,...,i_s\}} \in \mathbb{R}[x_1,...,x_n]$ a sum of squares polynomial (see also Sec.\ \ref{Sec:EsEaPositivstellen}).
Furthermore, the \textit{multiplicative monoid} $M$ of a set of polynomials $\{ g_1,...,g_k \}$ is defined by
\begin{equation}
M(g_1,...,g_l) = \prod_{i=1}^l g_i^{a_i},  \ a_i \in \mathbb{N}.
\end{equation}
With this definitions we can now state the Positivstellensatz \cite{Steng74}.
\begin{theorem}\label{positivstellen}
Let $\{f_i\}_{i=1,...s}, \ \{h_j\}_{j=1,...,m}, \ \{g_k \}_{k=1,...,l}$ be finite families of polynomials in $\mathbb{R}[x_1,...,x_n]$. Then the set $\left. \begin{cases}  & f_i(x) \geq 0, \ i= 1,...,s \\ x \in \mathbb{R}^n \mid & g_k(x) \neq 0, \ k=1,...,l \\  & h_j(x) = 0, \ j=1,...,m \end{cases} \right\}$ is empty iff
\begin{align}
\exists f \in P(f_1,...,f_s), & g\in M(g_1,...,g_l), h \in I(h_1,...,h_m) \nonumber \\ & \text{s.t.} \ f + g^2 + h = 0.
\label{pos}
\end{align}
\end{theorem}
Note that it is easy to see that if Eq.\ \eqref{pos} is fulfilled, there cannot exist a solution to the set of polynomial equations and inequalities, as for any solution $x_0 \in \mathbb{R}^n$ we have that $f(x_0) \geq 0$, $g^2(x_0)> 0$ and $h(x_0) = 0$. Thus, $f(x_0) + g^2(x_0) + h(x_0) >0$, which is in contradiction with Eq.\ \eqref{pos}.
For a proof of Theorem\ \ref{positivstellen} see \cite{Steng74}. The Positivstellensatz thus results in a single equation that corresponds to a necessary and sufficient condition for the existence of a real solution to a polynomial system of equations and inequalities. It is a very powerful theorem that leads in many cases to quite simple infeasibility certificates. In Sec.\ \ref{Sec:EsEaPositivstellen} we use the Positivstellensatz on the one hand to get certificates for when two given values of e.g. the source and the accessible entanglement cannot correspond to a physical state and on the other hand to show that the source entanglement together with its generalizations characterizes few-qubit bipartite entanglement. Whereas in the first case no functions $g_k$ are required, we will use them in the second. Note that by restricting the overall degree of the left-hand side of Eq.\ \eqref{pos} one can find solutions of this equation efficiently, as this problem can then be stated as a semidefinite program as shown in \cite{Parr03} (see Sec.\ \ref{Sec:EsEaPositivstellen}).

\section{Properties of bipartite entanglement}
\label{Sec:Prop}
In this section we present an operational characterization of bipartite pure state entanglement. Moreover, we discuss general properties of the source and accessible entanglement.
\subsection{Operational characterization of bipartite entanglement via the geometric measure of entanglement}
\label{Sec:OpChar}
In this section we show that the bipartite entanglement of a $d \times d$ system can be characterized as follows. Let $n = \lceil \log(d) \rceil$ and consider the bipartite state $\ket{\psi}_{A B}$ as a $2 n$-qubit state, i.e. $\ket{\psi}_{A_1...A_n B_1 ... B_n}$. As we will show below, the $d$ Schmidt coefficients of $\ket{\psi}_{A B}$ are given by the bipartite entanglement (measured with one function) of all possible bipartite splittings of the qubits in B versus the rest including the splitting A versus B. The entanglement is measured here with the geometric measure of entanglement, i.e. $E_g(\psi) = 1- \lambda_1$ with $\lambda_1$ the largest Schmidt coefficient. The geometric measure of entanglement\cite{Shim95} operationally quantifies entanglement by the distance of a state $\ket{\psi}$ to the nearest separable state. It is defined as $E_{g} = 1- \max_{\phi} \Vert \left< \phi | \psi \right> \Vert^2 $, with $\ket{\phi}$ an arbitrary product state. For bipartite pure states it is equivalent to the measure given above.  Note that in the multipartite case this measure has also been used in the context of quantum computing \cite{Gross09}.
\begin{lemma}
\label{2nx2nentanglement}
Let the bipartite state $\ket{\psi} = \sum_{i}  \sqrt{ \lambda_{i} } \ket{i}_A \ket{i}_B \in \mathbb{C}_A^{2^n} \otimes \mathbb{C}_B^{2^n}$ be considered as a  $2n$-qubit state $\ket{\psi} = \sum_{i_1,i_2,...,i_n} \sqrt{ \lambda_{i_1 i_2 ... i_n}} \ket{i_1 i_2 ... i_n}_A \ket{i_1 i_2 ... i_n}_B$. Then, $\ket{\psi}$ is uniquely (up to LUs) determined by the geometric measure of entanglement of all possible bipartite splittings with $k$ qubits in B versus the rest for all $k \in \{1,...n\}$. 
\end{lemma}
The proof of the above lemma is given in Appendix \ref{App:OpCharProof}. Note that in order to give this operational characterization, the basis, in which the $2n$-qubit state is written, has to be fixed, e.g. as in Lemma\ \ref{2nx2nentanglement}. This operational characterization of the bipartite entanglement is particularly interesting in case the bipartite entanglement of a multipartite qubit state is considered via a bipartite splitting.

\subsection{Properties of the source and accessible entanglement}
In this subsection we first show that the formula of the source entanglement can be simplified and then discuss general properties of the source and accessible entanglement.
\subsubsection{Simplification of the source entanglement formula}
\label{Sec:SimpliEs}
Here we want to show that by using the results from \cite{Post09} one can simplify the general formula of the source entanglement of bipartite pure states (see Eq.\ \eqref{sourceent}) that was introduced in \cite{Sauer15}. This result is stated in the following Lemma. 
\begin{lemma}
\label{LemmaEs}
The source entanglement of bipartite pure states (see \cite{Sauer15}) given in Eq.\ \eqref{sourceent} can be simplified to
\begin{eqnarray}
\label{SimplEs}
1 - \sum_{\sigma \in \Sigma_d} \frac{(\sum_{k=1}^d \sigma_k \lambda_k )^{d-1}}{\prod_{k=1}^{d-1} (\sigma_k - \sigma_{k+1})}.
\label{Essimp}
\end{eqnarray}
Hence, the source entanglement is a homogeneous function in $\lambda$ of degree d-1.
\end{lemma}
\begin{proof}
The proof is based on a result presented in \cite{Post09} concerning the divided symmetrization of a polynomial $f(x_1,...x_d)$, i.e. $\left< f \right> = \sum_{\sigma \in \Sigma_d} \sigma\left( \frac{f(x_1,...,x_d)}{\prod_{k=1}^{d-1} (x_k - x_{k+1})}\right)$. For all functions $f(x_1,...x_d)$ with degree smaller $d-1$ the divided symmetrization vanishes, i.e. $\left< f \right> = 0$. This can be easily seen by writing $\left< f \right> = \frac{g}{\Delta}$, as the Vandermonde determinant $\Delta = \prod_{1 \leq i<j \leq d} (x_i - x_j)$ is the common denominator of all the terms in the divided symmetrization. Note that $g$ has to be an antisymmetric polynomial, as the divided symmetrization is by definition symmetric and the Vandermonde determinant $\Delta$ is an antisymmetric polynomial. Furthermore, $\Delta$ is the antisymmetric polynomial of the smallest degree, deg $(\Delta) = \binom{d}{2}$. This is due to the fact that any antisymmetric polynomial vanishes if two of the variables are equal and thus, must have $x_i - x_j$ as a factor for all $i \neq j$. Now for any polynomial $f(x_1,...x_d)$ with deg $(f(x_1,...x_d)) < d-1$ the degree of the numerator of $\left< f \right>$ is smaller than the degree of the denominator. Therefore, the degree of the antisymmetric polynomial $g$ has to be smaller than the degree of the Vandermonde determinant, which, as explained above, is not possible. Hence, $g$ has to be equal to zero and therefore, the divided symmetrization vanishes, i.e. $\left< f \right> = 0$. Using this property of the divided symmetrization together with the binomial theorem for the numerator in the original formula of $E_s$ in Eq.\ \eqref{sourceent} one obtains directly the simplified version in Eq.\ \eqref{Essimp}.  
\end{proof}

As mentioned before the formula of the source entanglement (also the simplified version in Eq.\ \eqref{Essimp}) is valid for sorted Schmidt vectors $\lambda^{\downarrow}$. Investigating the properties of $E_s$ is often easier if the ordering is automatically fixed. Thus, we might want to change the variables in $E_s$ for certain problems and write the ordered Schmidt vector in terms of the extreme points of the convex set of sorted vectors, i.e.
\begin{eqnarray}
\lambda^{\downarrow} = \sum_{i=1}^d p_i e_i = M p,
\label{piSchmidt}
\end{eqnarray}
with $p = (p_1,...,p_d)$ and $M = \left( \begin{smallmatrix} 1 & 1/2 & 1/3 & \dotsb & 1/d \\ 0 & 1/2 & 1/3 & \dotsb & 1/d \\ 0 & 0 & 1/3 & ...& 1/d \\  0& 0 & 0&   & 1/d \\[-0.2cm] \vdots & \vdots  & \vdots  & & \vdots  \\  0 & 0 & 0 & ...& 1/d  \end{smallmatrix} \right)$. Using Eq.\ \eqref{piSchmidt} and the fact that the divided symmetrization of a polynomial vanishes, if the degree of the numerator is smaller than the degree of the denominator, leads to
\begin{eqnarray}
E_s(\psi) = 1 - \sum_{\sigma \in \Sigma_d} \frac{(\sum_{k=1}^{d-1} \frac{1}{k} p_k \sum_{l=1}^k \sigma_l)^{d-1}}{\prod_{k=1}^{d-1} (\sigma_k - \sigma_{k+1})}.
\label{Esp} 
\end{eqnarray}
Note that this formula does not depend on $p_d$, as $\sum_{l=1}^d \sigma_d = \frac{d(d+1)}{2}$ is a constant and thus, this term in the sum cancels, as the divided symmetrization of it vanishes (see above). 

\subsubsection{General properties of $E_s$ and $E_a$}
\label{Sec:GenProp}
Before investigating in which parameter region the values of $E_a$ and $E_s$ are lying, we summarize here some properties of these entanglement measures.  \par
Let us first show that the entanglement measures $E_a$ and $E_s$ are no entanglement monotones.
An entanglement monotone for pure states \cite{vidal1} is a function that is nonincreasing on average under LOCC, i.e. 
\begin{eqnarray}
E_{mon} ( \psi) \geq \sum_i p_i E_{mon} (\psi_i),
\label{Emon}
\end{eqnarray}
for any pure state ensemble $\{p_i, \ket{\psi_i} \}$ that is obtained from $\ket{\psi}$ via LOCC. A widely used feature of entanglement monotones for pure states is the fact that the convex roof construction \cite{horoRev} leads to entanglement measures on mixed states. Note that for the source and accessible entanglement such a construction is not necessary as they are already defined for any quantum state, including mixed states of any system size. Moreover, these two measures are defined in an operational way and thus, it is not surprising that instead of fulfilling condition\ \eqref{Emon} they only fulfill the physical LOCC monotonicity condition. For multipartite quantum states of four or more qubits it is easy to show that $E_s$ and $E_a$ are indeed increasing on average under LOCC, thus violating cond.\ \eqref{Emon}, by simply considering an isolated state $\ket{\psi}_{iso}$. Such a state can neither be reached nor transformed into any other state via LOCC. However, the state $\ket{\psi}_{iso}$ can be transformed with a non vanishing probability into a non-isolated state, $\ket{\psi_1}$. Denoting the other states in the ensemble of states into which $\ket{\psi}_{iso}$ is transformed by $\ket{\psi_i}$, we have $0 = E_s(\psi_{iso}) < p E_s(\psi_1) + \sum_i p_i E_s(\psi_i)$. A similar argument holds for the accessible entanglement. Note, however, that for e.g. pure three-qubit states within the W-class the source and accessible entanglement are indeed entanglement monotones. That is, the measures do not increase on average under LOCC for the W-class. Similar functions, i.e. entanglement monotones for states in the W-class, were also introduced in \cite{Chitambar12}. \par
Also for bipartite states one can easily construct a counter-example to show that cond.\ \eqref{Emon} is violated by the bipartite measures. In order to do so, consider a transformation of the state $\lambda_{\psi} = (0.6, 0.3, 0.1)$ into an ensemble containing the two states corresponding to $\lambda_{\psi_1} = (0.8, 0.15, 0.05)$ and $\lambda_{\psi_2} = (0.57, 0.32, 0.11)$ we obtain for the source and accessible entanglement respectively $E_s(\psi) = 0.63, \ E_a(\psi) = 0.36$ and $p_1 E_s(\psi_1) + p_2 E_s(\psi_2) = 0.631$, $ p_1 E_a(\psi_1) + p_2 E_a(\psi_2) = 0.37$. Thus, the bipartite source and accessible entanglement are no entanglement monotones\footnote{Note that as $E_s$ and $E_a$ are not SLOCC-invariant (see e.g. \cite{Gour10}), taking some root of these measures will also not lead to entanglement monotones. In the bipartite case one can also find simple examples verifying this. Furthermore, as $ E_s(\psi_{iso})= 0$ it is also clear for multipartite states that taking some root of $E_s$ will not lead to an entanglement monotone.}.
Note that also the Renyi-entropies, i.e. $S_{\alpha}(\rho) = \frac{1}{1-\alpha} \log(\tr(\rho^{\alpha})$, are no entanglement monotones for $\alpha>1$, $\alpha \neq \infty$ \cite{horoRev}, as they are not concave functions of the Schmidt coefficients. Hence, $S_{\alpha}$ cannot simply be generalized to mixed states by the convex roof construction (see e.g. \cite{horoRev}) for $\alpha>1$. However, the Renyi-entropies are Schur concave functions for all $\alpha$ and are thus valid entanglement measures for pure states. As mentioned before, in contrast to the Renyi-entropies the source and accessible entanglement are also defined for mixed states. \par
Let us now show that both the source and accessible entanglement are not additive on tensor products, i.e.
\begin{equation}
E_{s(a)}(\psi \otimes \phi) \neq E_{s(a)}(\psi) + E_{s(a)}(\phi).
\label{additive}
\end{equation}
The property of additivity is especially useful for entanglement measures if more copies of a state are considered. The source and the accessible entanglement are both unsurprisingly not additive. 
%Then their operational meaning given in the asymptotic regime of many copies of a state reduces to a single copy for additive measures. However, this property is not needed for the source and accessible entanglement in the single copy regime and thus, it is not surprising that $E_s$ and $E_a$ are not additive, as they have a clear operational interpretation also in the single copy regime. 
The non-additivity in Eq.\ \eqref{additive} can be easily proven by choosing two copies of the same 2-qubit state $\lambda_{\psi} = (\lambda_1, 1-\lambda_1)$.  We get for the source entanglement 
\begin{eqnarray}
E_s(\psi)  &=& 2 (1-\lambda_1), \nonumber \\
E_s^{4\rightarrow 2}(\psi) &=& 4(1-\lambda_1)^3 ,\\  \nonumber  E_s(\psi \otimes \psi) &=& 2 (1 - \lambda_1)^2 (1+ 2 \lambda_1(1+ 6 \lambda_1( 2\lambda_1 -1 ))),
\end{eqnarray}
where we chose the source entanglement of all $ 2 \times 2$ and all $4 \times 4$ states that can reach the $ 2 \times 2$ state $\ket{\psi}$ for the single copy case, such that we compare also the measures with the same dimension. However, for both measures the value of a single copy is not proportional to the value of two copies of the same state $\ket{\psi}$. For the accessible entanglement we obtain a similar result. \par
Another obvious property of the source and the accessible entanglement is that neither of them coincide with the entanglement of formation for pure states. As a matter of fact many known entanglement measures reduce to the entanglement of formation for pure states. Other examples of measures that do not have this property are the negativity \cite{negativity} or the robustness of entanglement \cite{robustness}. \par
Let us end this section with a brief discussion on the faithfulness of $E_s$ and $E_a$. A measure is faithful if it vanishes only on separable states and is strictly positive for entangled states. In the multipartite setting it is clear, that $E_s$ and $E_a$ cannot be faithful, due to the existence of isolated states. Another well known multipartite measure that is not faithful is the 3-tangle \cite{Coffman00}, which vanishes for all 3-qubit states in the W-class. In the bipartite case both $E_s$ and $E_a$ are, however, faithful. This can be easily shown as, on the one hand, for any mixed state, there exists a pure state that can reach this state, e.g. $\ket{\Phi^+}$ can be transformed into any state $\rho$, and on the other hand any entangled mixed state $\rho$ can be transformed into some other entangled state. In particular $\rho$ can be transformed into $\rho' = p \rho + (1-p) \proj{a,b}$, which is entangled for certain values of $p$ and product states $\ket{a,b}$, as the set of separable states is closed.

\section{Physical region of $E_s$, $E_a$ and other entanglement measures}
\label{physregionEsEa}
In this section we investigate the possible values $E_s$ and/or $E_a$ can take in terms of other measures. For measures that are given in terms of a polynomial in the Schmidt coefficients the Positivstellensatz (see Eq.\ \eqref{positivstellen}) gives the analytic solution to the problem. We investigate also the boundaries of these sets of states in terms of the entanglement measures and highlight in some figures the states in the source and accessible set in terms of the measures of some randomly chosen state. Furthermore for certain polynomial and some non-polynomial measures we plot the values the measures can have for low-dimensional bipartite states. \par
In \cite{BeSa03} a related problem has been considered. There, the authors investigate the optimization of functions of the form
\begin{equation}
S_f(\rho) = \sum_i f(\lambda_i),
\label{genEntr}
\end{equation}
given the value of some other function of this form. Note that the functions $S_f$ can be written as a sum over a function of a single Schmidt coefficient. Examples of such measures are entropy measures, e.g. the entanglement of formation given by \cite{Be}
\begin{equation}
S(\rho_A) = E_f(\psi)=- \tr(\rho_A \log(\rho_A)) = - \sum_i \lambda_i \log(\lambda_i),
\label{entropy}
\end{equation}
with $\rho_A = \tr_B(\proj{\psi})$.
The Schmidt vectors which are maximizing and minimizing the function $S_f(\rho)$ given $S_{f'}(\rho)$, respectively, have been shown in \cite{BeSa03} to be
\begin{align}\label{lambdaopt}
\{\lambda_i\}_{max} =&\{ \lambda_1, \lambda_0, ..., \lambda_0\},   \lambda_0 =\frac{1-\lambda_1}{d-1}\leq    \lambda_1   \\ \nonumber
\{\lambda_i\}_{min} =&\{ \lambda_1,..., \lambda_1, \lambda_0, 0, ...,0\}, \lambda_0=  1 - k \lambda_1, k= \lfloor 1/ \lambda_1 \rfloor.
\end{align}
Hence, both classes of states are parametrized by the single Schmidt coefficient $\lambda_1$. Note that in terms of the parameters $p_i$ (see Eq.\ \eqref{piSchmidt}) one can easily see that the states in Eq.\ \eqref{lambdaopt} correspond to convex combinations of exactly two extreme points of the convex set of sorted vectors. That is for $\{\lambda_i\}_{max}$ and $\{\lambda_i\}_{min}$ there are all but exactly two $p_i$`s equal to zero, i.e. $\{p_i\}_{max} = (p_1,0,...,0,p_d)$ and $\{p_i\}_{min} =\{0,0,...,p_k,p_{k-1},0,..,0\}$ with $k \in \{1,..,d\}$. It is important to note here that the results from \cite{BeSa03} do not apply for $E_s$ and $E_a$ as both measures involve products of Schmidt coefficients and can thus not be written as in Eq.\ \eqref{genEntr}. In spite of that, the states $\{\lambda_i\}_{max}$ and $\{\lambda_i\}_{min}$ play also important roles in our investigations as shown in the following sections. 

\subsection{Relation between source and accessible entanglement}
\label{Sec:RelEsEa}
Here we consider $3 \times 3$ and $4 \times 4$ bipartite quantum states (we skip $2 \times 2$ states, as for these states any measure depends only on a single Schmidt coefficient) and show plots for several polynomial measures. For the formulas of the plotted measures see Appendix A. In the second part of this section we consider similar plots but include the region of pairs of entanglement measures reachable by states in the source and accessible set (see Eq.\ \eqref{MaMs}). That is, we consider a random state $\ket{\phi}$  and illustrate which values of $E_s, E_a$ the states in the source and accessible set of $\ket{\phi}$ can have. 
\subsubsection{Bipartite $3 \times 3$ states}
Let us first consider two measures, that uniquely characterize the bipartite entanglement of $3 \times 3$ states as shown in \cite{Sauer15}, i.e. the source entanglement $E_s$ and the source entanglement of all $4 \times 4$ states, that reach a certain $3 \times 3$ state, i.e. $E_s^{4 \rightarrow 3}$ (see Appendix\ \ref{sec:AppendixA}). That is, given the value of $E_s, E_s^{4 \rightarrow 3}$ the state is (up to LUs) uniquely defined. \par
First note that $E_s$ and $E_s^{4 \rightarrow 3}$ can lead to a different order of $3 \times 3$ states for, i.e. for two states $\ket{\phi_1}$ and $\ket{\phi_2}$ we can have $E_s(\phi_1) > E_s(\phi_2)$ but $E_s^{4 \rightarrow 3}(\phi_1) < E_s^{4 \rightarrow 3}(\phi_2)$. This property is also observed with many other entanglement measures such as the negativity and the entanglement of formation, as it is a well known fact that different entanglement measures measure different aspects of entanglement and thus, it is not surprising that the order of the states can change for different measures. However, here the source entanglement and its generalization are very closely related and therefore, measure similar aspects of entanglement, but still they impose a different order on the states. 
Furthermore, the range of possible values for the pairs $(E_s, E_s^{4 \rightarrow 3})$ is surprisingly well confined, as can be seen in Fig.\ \ref{EsEs4}. For example, if $E_s(\psi) = 0.4$, $E_s^{4 \rightarrow 3}(\psi)$ can only have values between approximately $0.09$ and $0.2$. That is a state with some value for $E_s$ can only have values for $E_s^{4 \rightarrow 3}$ within a small range. When comparing other measures we will see that the range can be much larger than in this case. 
We also investigate the boundaries in Fig.\ \ref{EsEs4}. Let us first note that in all figures one can go on the boundary lines always from right to left via LOCC, i.e. from one random state on a boundary line we can obtain all other states to the left of the boundary deterministically. This can be easily checked for all boundaries, that we parametrized (see below). For all other boundaries we also checked this behavior numerically.
\begin{figure}[h!]
\centering
\includegraphics[width=0.4\textwidth]{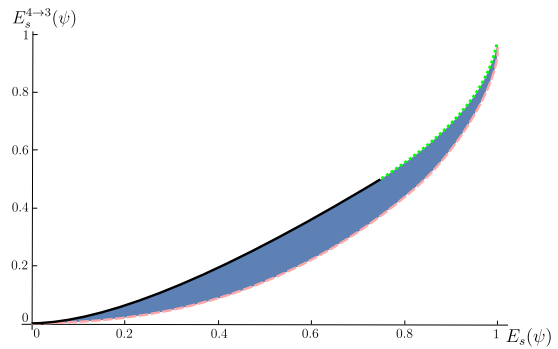}
\caption{Source entanglement of $3 \times 3$ states versus the source entanglement of all $4 \times 4$ states that reach the considered $3 \times 3$ state $\ket{\psi}$. Note that these two measures uniquely characterize $3 \times 3$ bipartite entanglement \cite{Sauer15}.}
\label{EsEs4}
\end{figure}
The boundaries in Fig.\ref{EsEs4} are indeed given by the states in Eq.\ \eqref{lambdaopt}. Interestingly, the dashed pink line corresponds to the states with $\lambda_2 = \lambda_3$, i.e. $\{\lambda_a\}_{max}= \{ \lambda_1, \frac{1-\lambda_1}{2}, \frac{1-\lambda_1}{2} \}$, the black line and the dotted green line are given by $\lambda_3 = 0$ and $\lambda_1 = \lambda_2$, respectively, i.e. $\{\lambda_b\}_{min}= \{ \lambda_1, 1-\lambda_1, 0 \}$, $\{\lambda_c\}_{min}= \{ \lambda_1, \lambda_1, 1-2 \lambda_1 \}$. That is these states do not only optimize a function $S_f(\rho)$ (see Eq.\ \eqref{genEntr}) given another function $S_{f'}(\rho)$, but also $E_s$ given $E_s^{4\rightarrow 3}$. This can also be easily seen using Fig.\ \ref{3x3volumes}, as the thick dashed lines that enclose the set of sorted Schmidt vectors, i.e. the boundary lines of the set of states excluding LU-equivalent states, correspond to the boundary lines in Fig.\ \ref{EsEs4}. As mentioned above these lines are parametrized by the convex combination of two extreme points of the set of sorted vectors. This can also be seen in Fig.\ \ref{3x3volumes}, where the dashed lines connect the vertices given by the separable state $(1,0,0)$, the maximally entangled $ 2 \times 2$ state $(1/2, 1/2,0)$ and the maximally entangled $3 \times 3$ state $(1/3, 1/3, 1/3)$. Thus, for example all states on the pink line in Fig.\ \ref{EsEs4} are given by the convex combination $p (1/3, 1/3, 1/3) + (1-p)(1,0,0)$, which is equivalent to the above defined $\{\lambda_a\}_{max}$. In the following we will see that when considering different measures the boundaries are no longer equivalent to the ones in Fig.\ \ref{3x3volumes}.  \par
The next measures we consider are the source and the accessible entanglement of a $3 \times 3$ state. Note that these two measures do not uniquely characterize $3 \times 3$ entanglement and the main  differences to the above compared measures are on the one hand that the range of values is now much broader and on the other hand that the boundaries are no longer solely given by the states in Eq.\ \eqref{lambdaopt}. 
\begin{figure}[h!]
\centering
\includegraphics[width=0.4\textwidth]{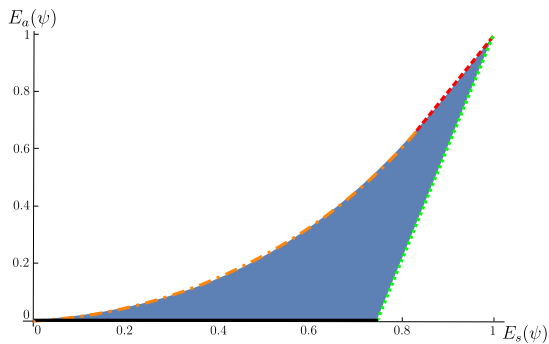}
\caption{Source entanglement versus the accessible entanglement of a $3 \times 3$ state $\ket{\psi}$.}
\label{3x3EsEa}
\end{figure}
In Fig.\ \ref{3x3EsEa} the dotted green and the black boundary correspond to the same states as in Fig.\ \ref{EsEs4}. In contrast to Fig.\ \ref{EsEs4} the states that maximize the accessible entanglement for a fixed value of the source entanglement are parametrized by $\lambda_2 = 1/3, \ 1/6 \leq \lambda_3 \leq 1/3$ for the dotted red line and $\lambda_2 = \sqrt{\lambda_3 ( 1- 2 \lambda_3)}, \ 0 \leq \lambda_3 < 1/6$ for the dotdashed orange line and furthermore, the pink line is no longer on the boundary. Hence, the states in Eq.\ \eqref{lambdaopt} are not enough to completely characterize the boundaries of Fig.\ \ref{3x3EsEa} and the boundaries are thus also not equal to the ones in Fig.\ \ref{3x3volumes}.   \par
Given the operational meaning of $E_s$ and $E_a$ it is now appealing to investigate what the entanglement properties (in terms of all these measures) are for the states that can be reached by and for the ones that can be transformed into one specific state $\ket{\phi}$. That is, we analyze given a state $\ket{\phi}$ with a certain value of $E_s$ and $E_a$, which pairs of values are in the source and accessible set of this state $\ket{\phi}$. In order to illustrate that, we include in Fig.\ \ref{3x3EsEs4MaMs}, \ref{3x3EsEaMaMs} the source and accessible entanglement of all states which are in the source or accessible set of $\ket{\phi}$. Let us first investigate the source entanglement and its generalization.
\begin{figure}[h!]
\centering
\includegraphics[width=0.4\textwidth]{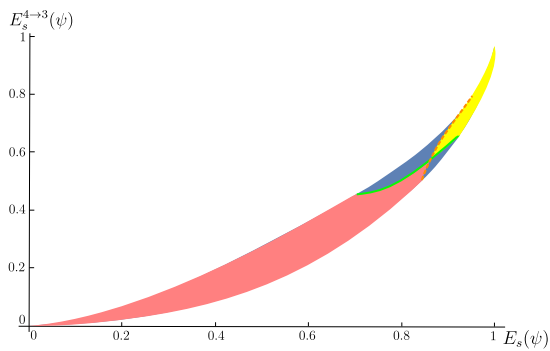}
\caption{$E_s$ versus $E_s^{4\rightarrow 3}$ including the source and the accessible set of a $3 \times 3$ state $\ket{\phi}$ with $\lambda_{\phi} = (0.52,0.28, 0.2)$ with boundaries.}
\label{3x3EsEs4MaMs}
\end{figure}
In the Fig.\ \ref{3x3EsEs4MaMs} the connected pink (lower left) area and the connected yellow (upper right) area correspond to the entanglement of the states in the accessible and the source set of a state $\ket{\phi}$, respectively. Hence, these sets include all states that can either be reached or can reach a certain state $\ket{\phi}$. The boundaries of the two sets, i.e. the green and the dashed orange line, are given by states with one of the two entanglement monotones $E_2(\psi), E_3(\psi)$ being equal to that of the state $\ket{\phi}$. More precisely, the green line corresponds to states with $E_2(\psi) = E_2(\phi)$ and the orange line is parametrized by states fulfilling $E_3(\psi) = E_3(\phi)$. Moreover, as we can see in Fig.\ \ref{3x3EsEs4MaMs} most states are LOCC-comparable, i.e. the state $\ket{\phi}$ can either reach or be reached by most of the states. We can conclude this from Fig.\ \ \ref{3x3EsEs4MaMs} as any point in the figure corresponds to a single state, due to the fact that $E_s$ and $E_s^{4\rightarrow 3}$ uniquely characterize the entanglement. However, there exist of course incomparable states. Interestingly, even though they uniquely characterize the entanglement of the states, there exist also incomparable states for which both $E_s$ and $E_s^{4 \rightarrow 3}$ are smaller but the states are not reachable by $\ket{\phi}$. The boundary lines of the pink and yellow set in Fig.\ \ref{3x3EsEs4MaMs} are as the boundary lines of the set of all states in Fig.\ \ref{EsEs4} equivalent to the ones in Fig.\ \ref{3x3volumes}, in which the source and accessible set of a state $\ket{\phi}$  is shown in terms of the Schmidt coefficients. Again this is only the case if we consider $E_s$ and its generalizations and we will see especially for the more involved $4 \times 4$ states that the boundaries of the pink and yellow sets are no longer simply parametrized by states with $E_i(\psi) = E_i(\phi)$ for some $i$ if we compare different measures. \par
Let us now consider the same state $\ket{\phi}$ and its source and accessible set respectively, as a function of $E_a$ versus $E_s$. As can be seen in Fig.\ \ref{3x3EsEaMaMs} there occur more boundaries. 
\begin{figure}[h!]
\centering
\includegraphics[width=0.4\textwidth]{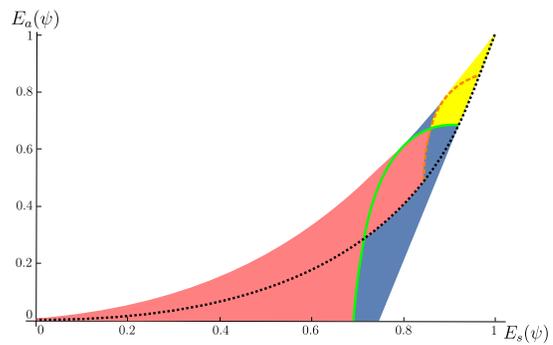}
\caption{$E_s$ versus $E_a$ including the source and the accessible set of a $3 \times 3$ state $\ket{\phi}$ with $\lambda_{\phi} = (0.52,0.28, 0.2)$ with boundaries.}
\label{3x3EsEaMaMs}
\end{figure}
In Fig.\ \ref{3x3EsEaMaMs} the green and dashed orange line are the same as in Fig.\ \ref{3x3EsEs4MaMs}. The dotted black line corresponds to states $\ket{\psi}$ with $\lambda_2 = \lambda_3$. The main difference to Fig.\ \ref{3x3EsEs4MaMs} is that an additional boundary, namely the dotted black line, is required. As these measures do not uniquely characterize the entanglement, one point in Fig.\ \ref{3x3EsEaMaMs} can correspond to different states. This is for example the case for the point in the pink accessible set where the dotted black and the green line intersect. The state on the green line is of the form $(\lambda_{\phi_1}, \lambda_2^{green}, \lambda_3^{green})$, whereas the state on the dotted black line is given by $(\lambda_1^{black}, \lambda_2^{black},\lambda_2^{black})$. \par
Interestingly, it seems to be a general feature (see also the $4 \times 4$ case), that for $E_s$ and its generalizations the boundaries are easy to determine, whereas for the plots including other measures, this is not the case. In the first case, the states corresponding to the boundaries are either of the same form as in Eq.\ \eqref{lambdaopt} or in case of the boundaries of the pairs corresponding to states in $M_s$ and $M_a$, i.e. the pink and yellow sets, have fixed values for the entanglement monotones $E_i$. We elaborate on that also in the next subsection, where $4 \times 4$ states are considered. \par

\subsubsection{Bipartite $4 \times 4$ states}
We consider here the case of $4 \times 4$ states. In addition to the investigations performed for $3 \times 3$ states, we compare here also the values obtainable for two 2-qubit states (viewed as $4 \times 4$ states, see Fig.\ \ref{4x4EsEa}).
Let us first consider the measures, which uniquely characterize the entanglement of 4-dimensional states, i.e. $E_s, E_s^{5\rightarrow 4}, E_s^{6\rightarrow 4}$ (as can be also seen with the help of the Positivstellensatz). The range of possible values is again quite constrained if we consider $E_s$ and its generalizations, whereas we get a wider range for values of the pair $(E_s,E_a)$. 
\begin{figure}[h!]
\centering
\includegraphics[width=0.4\textwidth]{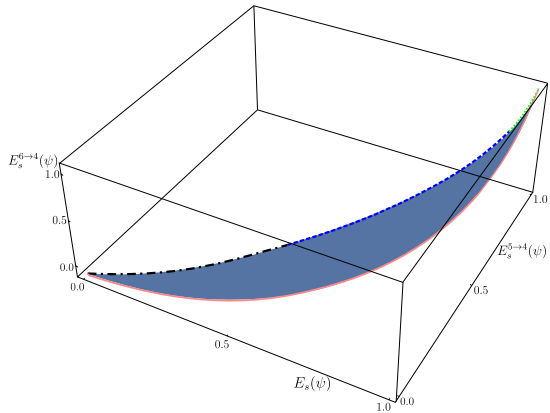}
\caption{$E_s$ versus its generalizations $E_s^{5\rightarrow 4}$ and $E_s^{6 \rightarrow 4}$.}
\label{4x4EsEs5Es6}
\end{figure}
Similar to the $3 \times 3$ case in Fig. \ref{4x4EsEs5Es6} the boundaries are again given by the states in Eq.\ \eqref{lambdaopt}. Thus, the pink line corresponds to the state $\{\lambda_a\}_{max}=\{ \lambda_1, \frac{1-\lambda_1}{3}, \frac{1-\lambda_1}{3},\frac{1-\lambda_1}{3} \}$ and the dotdashed black, the dashed blue and the dotted green line are given by the states $\{\lambda_b\}_{min}=\{ \lambda_1,1- \lambda_1,0,0\}$, $\{\lambda_c\}_{min}=\{\lambda_1,\lambda_1,1-2 \lambda_1,0\}$ and $\{\lambda_d\}_{min}=\{\lambda_1,\lambda_1,\lambda_1,1-3 \lambda_1\}$, respectively. Hence, the boundaries are again completely characterized by the states in Eq.\ \eqref{lambdaopt}, as it is for instance also the case considering only $E_s$ and $E_s^{5 \rightarrow 4}$. \par
Let us now, as in the $3 \times 3$ case, investigate the values of $E_s$ and its generalizations for states in the source and accessible set of a certain state $\ket{\phi}$.   
\begin{figure}[h!]
\centering
\includegraphics[width=0.4\textwidth]{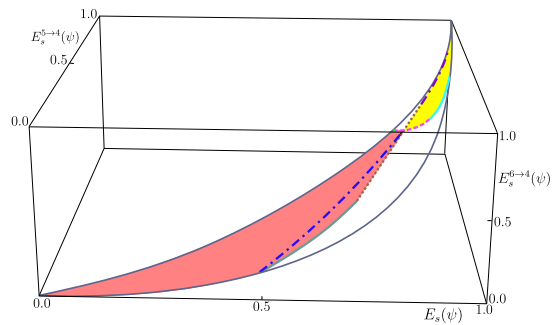}
\caption{$E_s$ versus $E_s^{5\rightarrow 4}$ versus $E_s^{6 \rightarrow 4}$ of the source (yellow) and accessible (pink) set of a state $\ket{\phi} = (0.4, 0.35, 0.2, 0.05)$. Note that we do not show all the states that are neither in the source nor the accessible set of $\ket{\phi}$.}
\label{4x4EsEs5Es6MaMs}
\end{figure}
To illustrate the result more clearly we show in Fig.\ \ref{4x4EsEs5Es6MaMs} only the values of the three measures for all states in either the source or accessible set of the state $\ket{\phi}$ and the boundaries (in light blue) of the set of all $4 \times 4$ states. The boundaries of the two sets are as in the $3 \times 3$ case given by states for which one of the entanglement monotones $E_i$ is fixed by the value it takes for the state $\ket{\phi}$. More precisely, the states on the gray line are parametrized by $E_4(\psi) = E_4(\phi)$ and $\lambda_2 = \lambda_3$, on the dotted brown line by $E_3(\psi) = E_3(\phi)$ and $E_4(\psi) = E_4(\phi)$, on the dotdashed blue line by $E_4(\psi) = E_4(\phi)$ and $\lambda_3 = \lambda_4$, on the dashed magenta line by $E_2(\psi) = E_2(\phi)$ and  $E_3(\psi) = E_3(\phi)$, on the green line by $E_2(\psi) = E_2(\phi)$ and $\lambda_4 = 0$, on the turquoise line by $E_2(\psi) = E_2(\phi)$ and $\lambda_3 = \lambda_4$ and on the dotdashed purple line by $E_4(\psi) = E_4(\phi)$ and $\lambda_1 = \lambda_2$. To summarize the states on the boundaries of the pink and yellow sets fulfill that at least one of the monotones $E_i$ is equal to the corresponding monotone of the state $\ket{\phi}$ and if only one fulfills this equality in addition two Schmidt coefficients of the states are equal to each other. Thus, these states are also similar to the ones given in Eq.\ \eqref{lambdaopt}. \par Similar to the $3 \times 3$ case the situation gets less transparent if we consider $E_s$ versus $E_a$. There, again the boundaries are not all given by those described above. \par  Note again that $E_s$ and $E_a$ do not uniquely characterize the entanglement of $4 \times 4$ states and thus, one point in the plot below can correspond to several different states. In the subsequent figure we also illustrate in orange the entanglement measures of $4 \times 4$ states, which are of the form $\ket{\Phi} \otimes \ket{\Psi}$, where both $\ket{\Phi}$ and $\ket{\Psi}$ are 2-qubit states. It is interesting to observe that these states almost optimize $E_a$ given $E_s$. Note, however, again that these measures do not uniquely define the states. The values of pairs $(E_a, E_s)$ for which there also exists a state $\ket{\psi} \neq \ket{\Phi} \otimes \ket{\Psi}$ can be read off combining Fig.\ \ref{4x4EsEa} with the figure for $4 \times 4$ states which is the equivalent to Fig.\ \ref{3x3EsEaPsuccto} for $3 \times 3$ states. Note that the red line is parametrized by all states that are given by two copies of a 2-qubit state, i.e.  $\ket{\Psi} \otimes \ket{\Psi}$.
\begin{figure}[h!]
\centering
\includegraphics[width=0.4\textwidth]{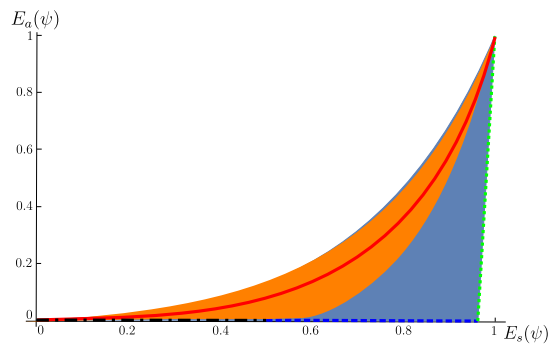}
\caption{Source entanglement versus the accessible entanglement of all $4 \times 4$ states $\ket{\psi}$.}
\label{4x4EsEa}
\end{figure}
The main difference of Fig.\ \ref{4x4EsEa} to the $3 \times 3$ case is that not all states in Eq.\ \eqref{lambdaopt} lie even on the boundary of the figure. That is the states on the pink line in Fig.\ \ref{4x4EsEs5Es6}, given by $\{\lambda_a\}_{max}=\{ \lambda_1, \frac{1-\lambda_1}{3}, \frac{1-\lambda_1}{3},\frac{1-\lambda_1}{3} \}$, are not optimizing $E_a(\psi)$ for a fixed value of $E_s(\psi)$. The other states from Eq.\ \eqref{lambdaopt} on the dotdashed black, the dotted green and the dashed blue line are still lying on the boundary (for a parametrization see above). The other boundary could of course be easily computed by maximizing $E_a(\psi)$ given $E_s(\psi)$ numerically \footnote{Note that one can of course also use Lagrange multipliers, however, this is cumbersome given the expression of $E_a$ (see Appendix \ref{sec:AppendixA}).}.

\subsection{Relation between source and accessible entanglement by the Positivstellensatz}
\label{Sec:EsEaPositivstellen}
In this section we review the idea of efficiently finding certificates for when a system of polynomial equations and inequalities has no solution in $\mathbb{R}$ (see \cite{Parr03}). We will show then how to use these certificates to find the possible region of two (or more) entanglement measures which are given as polynomial functions of the Schmidt coefficients. For instance, this method can be used to determine analytically all possible pairs $(E_s, E_a)$, which are accessible by a state. \par
It has been shown in \cite{Parr03} that one can efficiently find these certificates by using the Positivstellensatz and fixing the overall degree of Eq.\ \eqref{pos}, as then the problem can be written as a semidefinite program (SDP).  Hence, by solving these SDPs and obtaining certificates we can completely solve the problem of when certain values of  the source entanglement and its generalizations correspond to a physical state or not. Note that this would not be possible for measures $S_f(\rho)$ as in Eq.\ \eqref{entropy}, as they are no polynomials. For each point in the figures lying outside the boundaries of the blue sets in all previous figures we find a certificate that tells us, that there exists no state having these values for the entanglement measures. \par To explain the idea of obtaining these certificates by solving a SDP, we first recall the definition of sum of squares (SOS) polynomials and the fact that the existence of a SOS decomposition for polynomials can be decided by solving a SDP feasibility problem (see \cite{Parr03} and references therein). The computational tractability of the SOS polynomials together with fixing the overall degree of Eq.\ \eqref{pos} leads then also to the relaxations in the Positivstellensatz as explained below. \par 
A SOS polynomial $F(x), \ x \in \mathbb{R}^m$ is a real-valued polynomial of even degree, that can be written in the form $F(x) = \sum_i f_i(x)^2$, with $f_i(x) \in \mathbb{R}[x]$. Clearly not every nonnegative polynomial is SOS \cite{Hilbert}, e.g. the Motzkin form $M(x,y,z) = x^4 y^2 + x^2 y^4 + z^6 - 3 x^2 y^2 z^2$ is nonnegative but cannot be written as a SOS. 
A polynomial $F(x)$ of degree $2d$ is SOS iff it can be written as a quadratic homogeneous polynomial in $z$, where the vector $z$ contains all monomials of $x$ with degree less or equal to $d$, i.e.
\begin{equation}
F(x) \ \text{is SOS} \Leftrightarrow  \exists Q \geq 0: \ F(x) = z^T Q z,
\label{SOS}
\end{equation}
where $z = [1, x_1, x_2,...,x_m, x_1 x_2,...,x_m^d]$ and the positive semidefinite $m \times m$ matrix $Q$ is constant. Using this decomposition of SOS polynomials it is easy to restate the problem of deciding whether a polynomial is SOS or not as a SDP feasibility problem, i.e.
\begin{align}\label{SDPSOS}
\text{find} \ \ &Q \nonumber\\
\text{subject to}\ \  &\tr(A_{\alpha} Q) = c_{\alpha} \\
&Q \geq 0.\nonumber
\end{align} 
The linear equality constraints in this feasibility problem are derived from $F(x) = z^T Q z$ by comparing the coefficients in these polynomials. Note that any SDP feasibility problem possesses a dual problem which gives a witness, proving that the primary problem has no solution. In the case investigated here, the potential witness can be found by writing $F(x) = z^T Q z = \tr(z z^T Q)$ and replacing $z z^T$ by a matrix $W$, that fulfills the same linear relations among its entries as $z z^T$ (relaxing the condition of having rank 1). As long as $Q$ represents the original polynomial $F(x)$ the $\tr(W Q)$ does not depend on the specific choice of $Q$. Then the dual problem is equal to
\begin{align}\label{SDPSOSdual}
\text{find} \ \ &W \nonumber\\
\text{subject to}\ \  &\tr(W Q) < 0 \\
&W \geq 0\nonumber \\
& w_{i j} = w_{k l} \ \text{for} \{(i,j),(k,l)\} \in I\nonumber,
\end{align}
where $I$ is chosen such that the $W$-matrix fulfills the same linear conditions among the entries as $z z^T$ \footnote{Note that even though there are more conditions on the entries of $z^T z$, only the conditions that are linear in the entries of $W$ are imposed on $W$.}. For instance for $z = [x_1^2, x_2^2, x_1 x_2]$ the entries of $W$ have to fulfill $w_{12} = w_{33}$. \par
Using these results on SOS polynomials we now review the method of finding bounded-degree certificates from the Positivstellensatz. 
In order to do so, the overall degree $d_0$ of Eq.\ \eqref{pos}, i.e. $\text{deg}(f+ g^2 + h) = d_0$ is fixed \cite{Parr03}. Then the polynomial $g$, which is generated by the multiplicative monoid of the set $\{g_k\}_{k=1}^t$, is either equal to 1 if $t=0$ or it is given by $g= \prod_{i=1}^t g_i^m$, with $m$ chosen such that the degree of $g^2$ is less than or equal to $d_0$. The polynomial $f$ in the cone of the set of inequalities $\{f_i\}_{i=1}^s$ is parametrized by $f = s_0 + s_1 f_1 +...+s_s f_s + s_{12} f_1 f_2 +...+s_{12...s} f_1...f_s$, with $s_{\{i_1,...,i_s\}}$ SOS polynomials of degree less than or equal to $d_0$. Furthermore, the polynomial $h$ in the ideal of the set $\{h_j\}_{j=1}^m$ is equal to $h = t_1 h_1 + ... + t_m h_m$ with some polynomials $t_i$ of degree again less than or equal to $d_0$. The corresponding SDP feasibility problem is then given by
\begin{align} \label{SDPPOS}
\text{find} \ \ &Q_{\{i_1,...,i_s\}},  \nonumber\\ 
\text{subject to}\ \  &f + g^2 + h=0 \\
&Q_{\{i_1,...,i_s\}} \geq 0,\nonumber
\end{align} 
with $s_{\{i_1,...,i_s\}} = z_{\{i_1,...,i_s\}}^T Q_{\{i_1,...,i_s\}} z_{\{i_1,...,i_s\}}$ and the monomial vector $z_{\{i_1,...,i_s\}} =  [1, x_1, x_2,...,x_m, x_1 x_2,...,x_m^k]$, with $k \leq d_0/2$. Here, $k$ is chosen for each $z_{\{i_1,...,i_s\}}$, such that the overall degree of Eq.\ \eqref{pos} is equal to $d_0$. Note that the equation $f + g^2 + h = 0$ leads to the equality constraints for the SDP feasibility problem in \eqref{SDPPOS}. 
Hence, if the above SDP problem is feasible for a fixed degree $d_0$ of the left-hand side of Eq.\ \eqref{pos} we automatically get a certificate telling us that the set of solutions of polynomial inequalities and equations in Theorem \ref{positivstellen} is empty. Thus, we can use these certificates to verify all plots from the previous section. We simply have to show that there exist infeasibility certificates for all points in the plot, that do not correspond to a physical state. The certificates can be found numerically by using the software package SOSTOOLS \cite{SOS}. \par
To illustrate the method we consider now one particular example, namely the case of $3 \times 3$ states presented in Fig.\ \ref{EsEs4} with the two measures $E_s$ and $E_s^{4 \rightarrow 3}$. First we choose some numerical values $E_s^0$ and $(E_s^{4 \rightarrow 3})^0$ for which we want to check the existence of a certificate, telling us that no physical state corresponds to these values. That is we consider the following system of polynomial equations and inequalities 
\begin{equation}
\left. \begin{cases}  & 1- q_1^2 - q_2^2 \geq 0 \\ \left(
\begin{smallmatrix}
q_1\\
q_2\\
\end{smallmatrix}
\right) \in \mathbb{R}^2 \mid & E_s(\vec{q}) - E_s^0 = 0 \\  &  E_s^{4 \rightarrow 3}(\vec{q}) -(E_s^{4 \rightarrow 3})^0=0\end{cases} \right\},
\label{PositivstellenEsEs4}
\end{equation}
with $q_i^2 = p_i$ and the $p_i$'s are the components of the Schmidt vector that is given by the extreme points of the convex set of sorted vectors, see Eq.\ \eqref{piSchmidt}. Note that we write the measures here in terms of the $p_i$'s, such that we have only a single inequality in the set \eqref{PositivstellenEsEs4}, as we do not have to take into account the ordering of the $\lambda_i$'s. Furthermore, we substitute the parameters $p_i$ with $q_i^2$, such that we do not need to impose the condition $p_i \geq 0$. Note that the inequality in the set \eqref{PositivstellenEsEs4} is due to the norm of the Schmidt vector, i.e. $p_3 = q_3^2 = 1 - q_1^2 - q_2^2$. From the Positivstellensatz we know that the above set is empty, iff the equation 
\begin{align}\label{certificateEsEs4}
s_0 + s_1 (1- q_1^2 - q_2^2) +& t_1 (E_s(\vec{q}) - E_s^0 ) +\\ \nonumber& t_2 (E_s^{4 \rightarrow 3}(\vec{q}) -(E_s^{4 \rightarrow 3})^0) + 1 = 0
\end{align}
is fulfilled for some SOS-polynomials $s_0, s_1$ and some arbitrary polynomials $t_1, t_2$ in the parameters $q_1, q_2$ and with arbitrary degree. The existence of polynomials fulfilling Eq.\ \eqref{certificateEsEs4} can then be verified using the software package SOSTOOLS by fixing the degree of the whole equation. In this case we find certificates for all values $E_s^0$ and $(E_s^{4 \rightarrow 3})^0)$ that do not belong to a bipartite $3 \times 3$ state for the lowest possible degree of Eq.\ \eqref{certificateEsEs4}, which is given by $deg(E_s^{4 \rightarrow 3}(\vec{q})) = 6$. More precisely, by solving a SDP problem for $d_0 = 6$ we find SOS-polynomials with $deg(s_0) = 6$, $deg(s_1) = 4$ and general polynomials of degree $deg(t_1) = 2$ and $deg(t_2)=0$, see Eq.\ \eqref{certificateEsEs4}. Note that we do not only find low degree certificates for this simple case, but also for the more involved cases of bipartite $4 \times 4$ states, which is done in exactly the same way as explained here.
Furthermore, for certain values of $E_s^0$ and $(E_s^{4 \rightarrow 3})^0$ we can simply read off Eq.\ \eqref{certificateEsEs4}, that there do not exist polynomials $s_0$, $s_1$, $t_1$ and $t_2$ fulfilling the equation for any degree and thus, these values correspond to physical states.  That is, for states lying on the boundaries in Fig.\ \ref{EsEs4}, i.e. the states in Eq.\ \eqref{lambdaopt}, $E_s^0$ and $(E_s^{4 \rightarrow 3})^0)$ have no constant terms. Thus, the only constant terms in Eq.\ \eqref{certificateEsEs4} are given by the constant terms in $s_0 = z_0^T Q_0 z_0$ and $s_1 = z_1^T Q_1 z_1$ and 1. The constant terms in the each of the two SOS-polynomials are equal to the first matrix entry of $Q_0$ and $Q_1$, i.e. $(Q_0)_{11}$ and $(Q_1)_{11}$ and hence, the condition on the constant terms in Eq.\ \eqref{certificateEsEs4} is given by $(Q_0)_{11} + (Q_1)_{11}+ 1 = 0$. As both of these matrices have to be positive semidefinite we have $(Q_0)_{11} \geq 0, \ (Q_1)_{11} \geq 0$. Thus, the condition on the constant terms cannot be fulfilled and therefore, for the states on the boundaries given by Eq. \eqref{lambdaopt} there exists certainly no certificate as in Eq. \eqref{certificateEsEs4}. Hence, there exists a solution to $E_s = E_s^0$ and $E_s^{4 \rightarrow 3} = (E_s^{4 \rightarrow 3})^0$. 

\subsection{Relation between source/accessible entanglement and other entanglement measures}
\label{Sec:RelEsEaotherE}
In this section we investigate possible values of the source and accessible entanglement, respectively, given the value of some other entanglement measure. The measures we consider are on the one hand the entanglement of formation \cite{Wootters1} (see Eq.\ \eqref{entropy}), as it is an important bipartite entanglement measure with a clear operational meaning in the asymptotic limit. Thus, it is from an operational viewpoint very different to the source/accessible entanglement, which are defined in the single-copy scenario. Therefore it is very interesting to compare these measures.
On the other hand we consider the negativity \cite{negativity}, which is a mathematically tractable entanglement measure that is widely used. The negativity reads
\begin{equation}
E_N(\psi) = \frac{\Vert \rho^{T_A}\Vert_1-1}{d-1} = \frac{2}{d-1} \sum_{i<j} \sqrt{\lambda_i \lambda_j},
\end{equation} 
where  $\rho^{T_A}$ denotes the partial transpose with respect to system A and $\Vert \cdot \Vert_1$ the trace norm.
Note that the behavior for both the negativity and the entanglement of formation in terms of $E_s$ or $E_a$ is very similar. However, for $E_f$ the Positivstellensatz cannot be used to analytically verify the plots below, as this measure is not a polynomial function of the Schmidt coefficients. Note, however, that the negativity can easily be written as a polynomial, if we consider instead of the parameters $\sqrt{\lambda_i}$ $\lambda_i$. \par
Let us start again with bipartite $3 \times 3$ dimensional states and consider $4 \times 4$ states afterwards.
\subsubsection{Bipartite $3 \times 3$ states}
First, we investigate the values the entanglement of formation can take in terms of the source or the accessible entanglement.
\begin{figure}[h!] 
\begin{minipage}[h!]{3cm}
\centering  \hspace*{-0.9cm}
\includegraphics[width=1.3\textwidth]{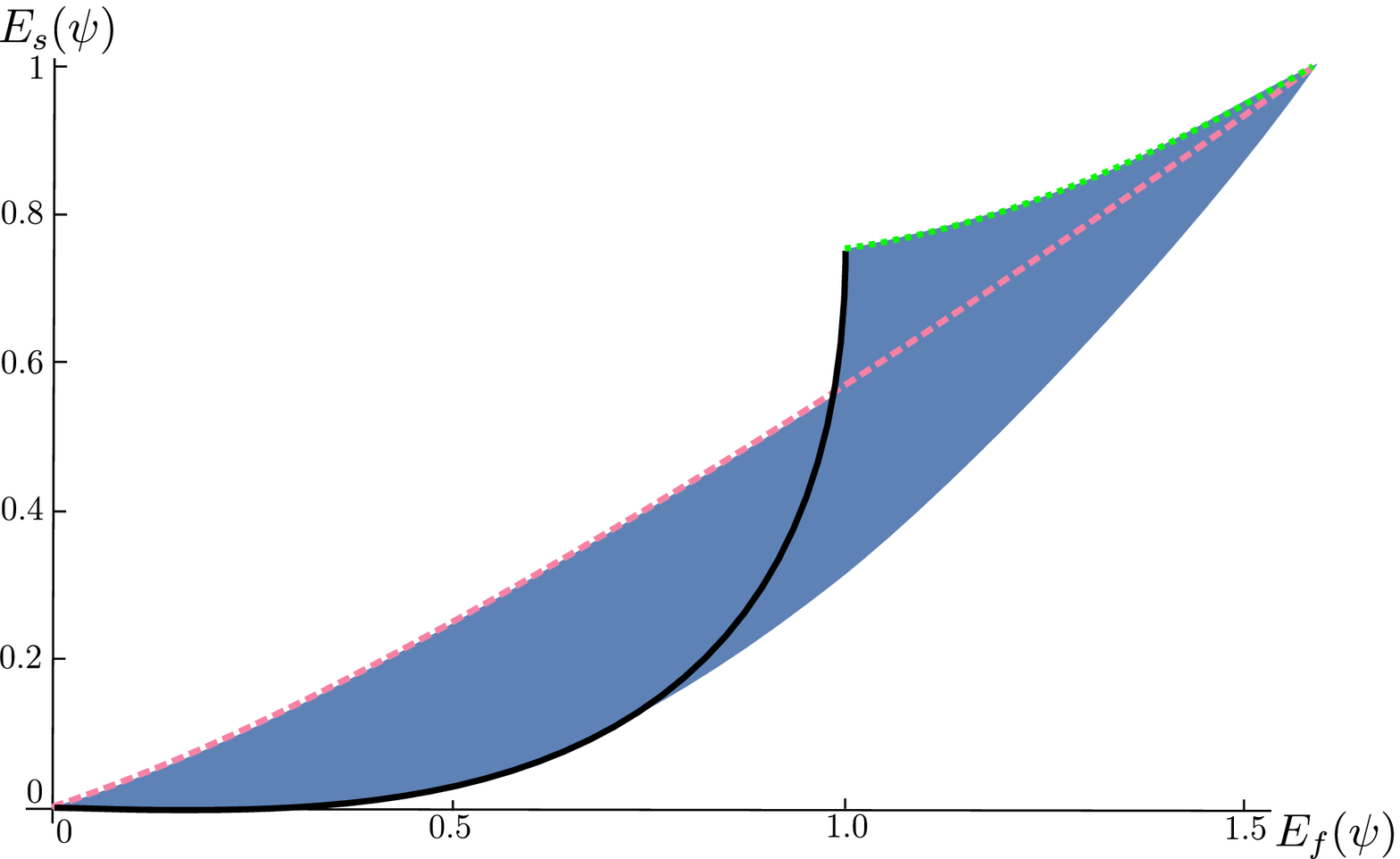}
\caption{Source entanglement versus entanglement of formation.}
\label{3x3EsEf}
\end{minipage}
\hfil
\begin{minipage}[h!]{3cm}
\centering
\includegraphics[width=1.3\textwidth]{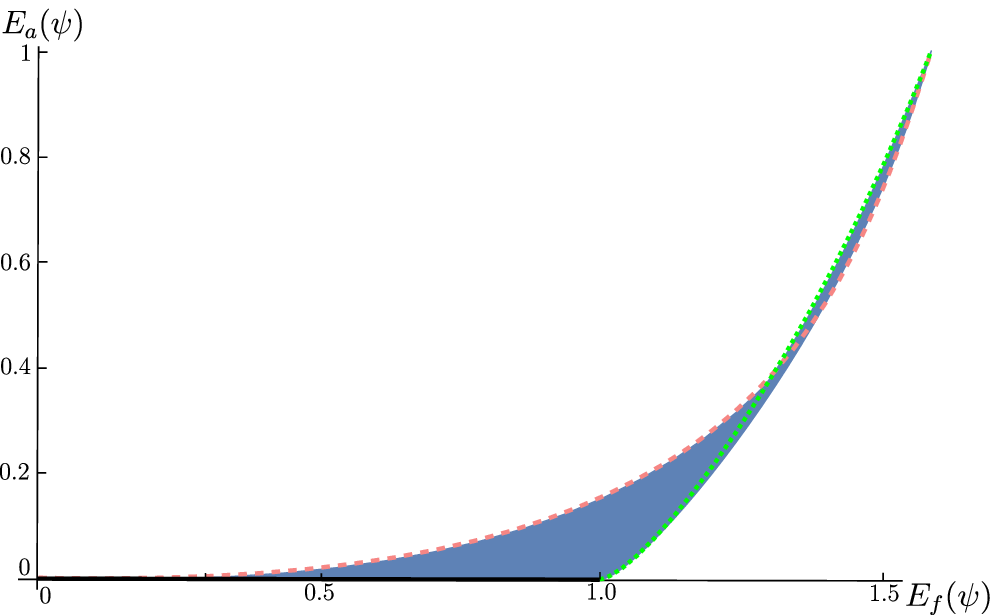}
\caption{Accessible entanglement versus entanglement of formation. }
\label{3x3EaEf}
\end{minipage}
\end{figure}
For all these plots (note that these measures also do not uniquely characterize the entanglement of $ 3 \times 3$ bipartite states) we find similar results for the boundaries of the set of physical states as for the plots of $E_s$ versus $E_a$. In particular, the states in Eq.\ \eqref{lambdaopt} are not enough to completely define the boundaries. In Fig.\ \ref{3x3EsEf} the dotted green, the dashed pink, and the black line are given by states of this form (see also Fig.\ \ref{3x3EsEa}).  
Interestingly, there is a "jump" in the possible values for $E_s$ around $E_f \approx 1$. The reason for this is that the states on the black line are actually $ 2 \times 2$ dimensional states and close to $E_f \approx 1$ the states are close to the maximally entangled state in $ 2 \times 2$ dimensions. These states cannot be reached by many $ 3 \times 3$ bipartite states via LOCC, thus the source entanglement is relatively high for these states, whereas the entanglement of formation cannot be larger than 1. Moreover, there are $ 3 \times 3$ states whose entanglement of formation is also around 1 but which can be reached by many states, thus the source entanglement is relatively small. This explains the large range of possible values of $E_s$ for $E_f \approx 1$.  The same feature is depicted in Fig.\ \ref{3x3EsEa}, where the possible range of values for $E_s$ is largest for $E_a = 0$, which includes all $ 2 \times 2$ dimensional states, that cannot access any $ 3 \times 3$ dimensional state by LOCC. This fact is also very obvious in Fig.\ \ref{3x3EaEf}, where $E_a$ versus $E_f$ is depicted. Note though, that the range of possible values of the pairs $(E_a, E_f)$ is much more constrained than the range of values for $(E_s, E_f)$.

\begin{figure}[h!] 
\begin{minipage}[H]{3cm}
\centering \hspace*{-2.2cm}
\includegraphics[width=1.8\textwidth]{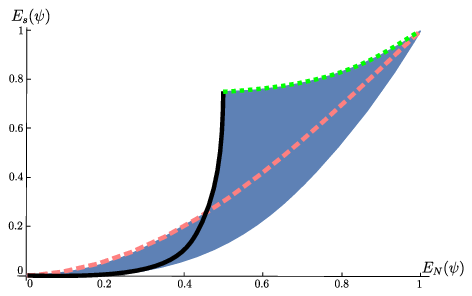}
\caption{Source entanglement versus negativity for $3 \times 3$ states.}
\label{3x3EsEn}
\end{minipage}
\hfil
\begin{minipage}[H]{3cm}
\centering
\includegraphics[width=1.3\textwidth]{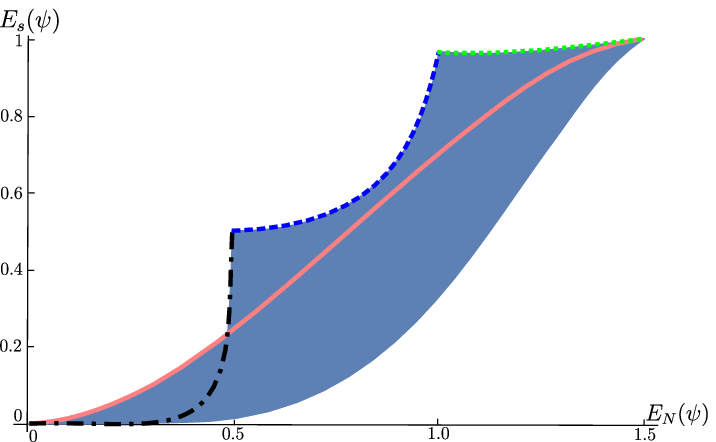}
\caption{Source entanglement versus negativity for $4 \times 4$ states.}
\label{4x4EsEn}
\end{minipage}
\end{figure}

Furthermore, the dashed pink, the dotted green and the black boundaries are given by the states from Eq.\ \eqref{lambdaopt} and there is again a small part of the boundary missing, that could be obtained by numerically optimizing $E_f$ given $E_a$.  \par
As mentioned before we also want to investigate the negativity in terms of the source entanglement. The behavior of the negativity is very similar to the entanglement of formation, as can be seen in Fig.\ \ref{3x3EsEn}. 
The dotted green, dashed pink and black line are again given by the states in Eq.\ \eqref{lambdaopt} and the missing boundary can only be obtained by numerical optimization. Around $E_N \approx 0.5$ the range of $E_s$ is also quite large, as all $ 2 \times 2$ states, which are close to the maximally entangled state, have approximately this value for the negativity. \par
We also consider as an example the values of $E_f$ and $E_s$ all states in the source and accessible set of a certain state $\ket{\phi}$ can have. For this we find again similar results as in Fig.\ \ref{3x3EsEaMaMs}. 
\begin{figure}[h!] 
\begin{minipage}[H]{3cm}
\centering \hspace*{-0.9cm}
\includegraphics[width=1.3\textwidth]{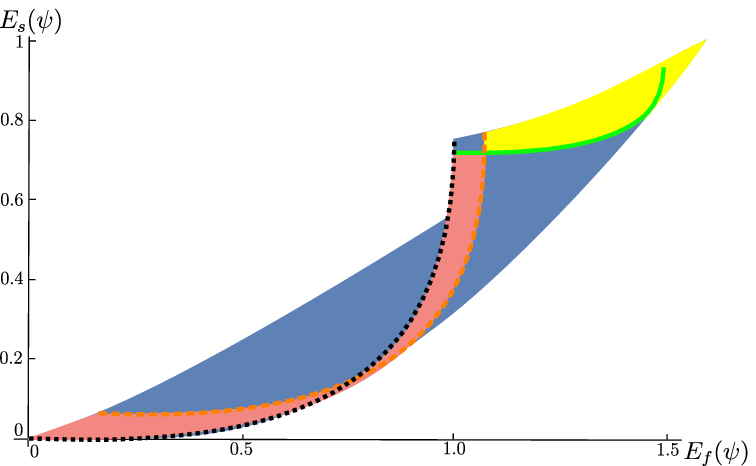}
\caption{Source entanglement versus entanglement of formation highlighting the states in the source and accessible set of a state $\ket{\phi}$ with the Schmidt vector $(0.51, 0.48, 0.01)$. }
\label{3x3EsEfMaMs}
\end{minipage}
\hfil
\begin{minipage}[H]{3cm}
\centering
\includegraphics[width=1.3\textwidth]{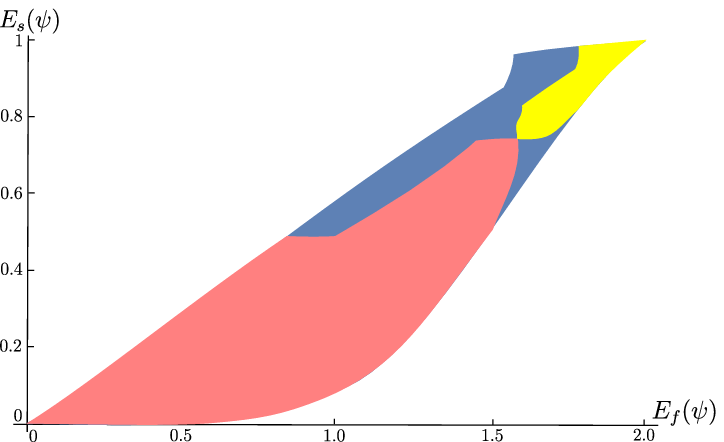}
\caption{Source entanglement versus entanglement of formation highlighting the states in the source and accessible set of a state $\ket{\phi}$ with the Schmidt vector $(0.51, 0.4, 0.1, 0.05)$.  }
\label{4x4EsEfMaMs}
\end{minipage}
\end{figure}

The dashed orange and green boundaries of the pink (lower left) accessible set and the yellow (upper right) source set in Fig.\ \ref{3x3EsEfMaMs} are as in Fig.\ \ref{3x3EsEaMaMs} given by states, for which one of the entanglement monotones $E_i$ is fixed, i.e. $E_i(\psi) = E_i(\phi)$. Moreover, the dotted black boundary is parametrized by $\lambda_3=0$. This shows that also in this case it is not enough to have just the dashed orange and green boundary, as it would be for the plot of $E_s$ and $E_s^{4 \rightarrow 3}$. \par
Note that the behavior of $E_f$ in terms of $E_N$ is very similar to $E_s$ and its generalizations. That is all boundaries of the set of possible values of pairs $(E_f, E_N)$ are given by the states in Eq.\ \eqref{lambdaopt} in contrast to the figures shown above. 

\subsubsection{Bipartite $4 \times 4$ states}
Let us now also investigate these measures for $4 \times 4$ states. We get again similar results as for the $ 3 \times 3$ case. 
In Figs. \ref{4x4EsEf}, \ref{4x4EaEf} and \ref{4x4EsEn} the lines are parametrized by the states in Eq.\ \eqref{lambdaopt}. More precisely, the pink line corresponds to the state $\{\lambda_a\}_{max}=\{ \lambda_1, \frac{1-\lambda_1}{3}, \frac{1-\lambda_1}{3},\frac{1-\lambda_1}{3} \}$ and the dotdasehd black, the dashed blue and the dotted green lines are given by the following states $\{\lambda_b\}_{min}=\{ \lambda_1,1- \lambda_1,0,0\}$, $\{\lambda_c\}_{min}=\{\lambda_1,\lambda_1,1-2 \lambda_1,0\}$ and $\{\lambda_d\}_{min}=\{\lambda_1,\lambda_1,\lambda_1,1-3 \lambda_1\}$, respectively. Note that the missing boundaries can be found numerically. Furthermore, it can also be seen in these figures, that the range of values for either $E_s$ or $E_a$ is largest for the values that are close to either the maximally entangled $ 2 \times 2$ or $ 3 \times 3$ state. In Fig.\ \ref{4x4EsEn} the maximally entangled $ 2 \times 2$ and $ 3 \times 3$ state lie both on the boundary, whereas in Fig. \ref{4x4EsEf} only the $ 3 \times 3$ state lies on the boundary (and thus there is only one "jump" in this figure).  In Fig.\ \ref{4x4EsEfMaMs} we highlight again the values for $E_s$ in terms of $E_f$ for all states that can either be accessed (pink lower left set) or reach (yellow upper right set) a certain state $\ket{\phi}$. Note again that in Fig.\ \ref{4x4EsEfMaMs} the boundaries of the two sets can not be as easily obtained as in Fig.\ \ref{4x4EsEs5Es6MaMs}. 

\begin{figure}[h!] 
\begin{minipage}[H]{3cm}
\centering \hspace*{-0.9cm}
\includegraphics[width=1.3\textwidth]{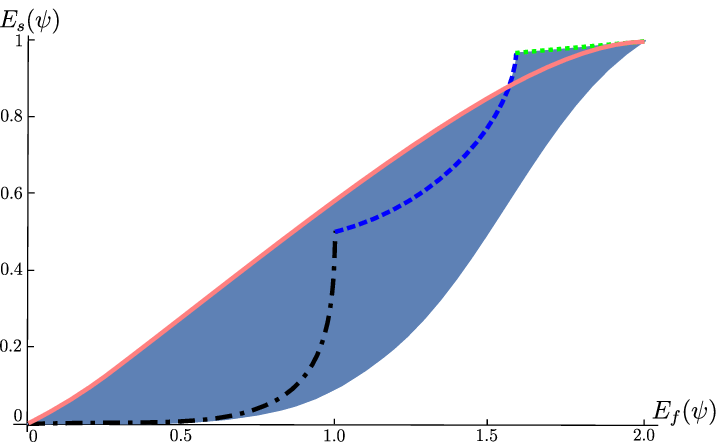}
\caption{Source entanglement versus entanglement of formation. }
\label{4x4EsEf}
\end{minipage}
\hfil
\begin{minipage}[H]{3cm}
\centering
\includegraphics[width=1.3\textwidth]{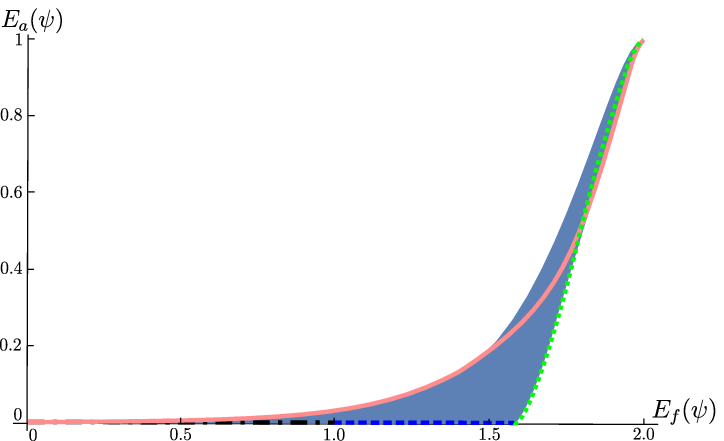}
\caption{Accessible entanglement versus entanglement of formation. }
\label{4x4EaEf}
\end{minipage}
\end{figure}

\subsection{Relation between source and accessible entanglement in connection with the success probability}
\label{Sec:Psucc}
In this section we consider probabilistic transformations \cite{vidal2} of states via LOCC and relate them to the source and accessible entanglement. The optimal success probability for these transformations is given by \cite{vidal2}
\begin{equation}
P(\Psi \rightarrow \Phi) = \min_{i \in [1, d]} \frac{E_i(\Psi)}{E_i(\Phi)}.
\label{Psucc}
\end{equation}
It has been shown in \cite{vidal2} that $\ket{\Psi}$ can be converted into $\ket{\Phi}$ with success probability p iff $p \lambda_{\Phi} \stackrel{W}{\prec} \lambda_{\Psi}$, where the weak majorization has to be considered as the sum over all vector entries is no longer equal for these vectors. 
 Here, we mention a simple observation in connection with the success probability $P(\Psi \rightarrow \Phi)$ corresponding to the optimal transformation $\ket{\Psi}\rightarrow \{p \ket{\Phi},p_j \ket{\Phi_j}\}$ and entanglement of the states $\ket{\Phi_j}$ in the ensemble.
\begin{observation}
Let $\ket{\Psi}\rightarrow \{p \ket{\Phi},p_j \ket{\Phi_j}\}$ be the optimal LOCC protocol to reach $\ket{\Phi}$ and let \bea P(\Psi\rightarrow \Phi)={\mbox min}_k \frac{E_{k} (\Psi)}{E_{k} (\Phi)}=\frac{E_{k_0} (\Psi)}{E_{k_0} (\Phi)}.\eea
Then, $E_{k_0} (\Phi_j)=0$ for any $\Phi_j$ that occurs in the ensemble. In particular the smallest $d-k_0$ Schmidt coefficients vanish for all $\Phi_j$ that occur in the ensemble.
\end{observation} 

For instance in the extreme case, where $k_0=2$ all the states $\ket{\phi_j}$ must be product states. That is, all the entanglement is transformed into the state $\ket{\Phi}$.
\begin{proof}
As shown in \cite{Plenio99} the transformation $\ket{\Psi}\rightarrow \{p_1,\ket{\Phi},p_j \ket{\Phi_j}\}_{j>1}$ is possible iff $ E_k(\Psi)\geq \sum_j p_j E_k(\Phi_j)$. From the last inequality we have that
\bea p_1 \leq \frac{E_{k} (\Psi)-c_k}{E_{k} (\Phi)} \ \ \forall k,\eea where $c_k=\sum_{j>1} p_j E_k(\Phi_j)$ and therefore
\bea p_1 = {\mbox min}_k \frac{E_{k} (\Psi)}{E_{k} (\Phi)} =\frac{E_{k_0} (\Psi)}{E_{k_0} (\Phi)}\leq \frac{E_{k_0} (\Psi)-c_{k_0}}{E_{k_0} (\Phi)} .\eea As $c_{k_0}\geq 0$ we obtain that $c_{k_0}= \sum_{j>1} p_j E_{k_0}(\Phi_j)=0$. Hence, the smallest $d-k_0$ Schmidt coefficients of all states $\ket{\Phi_j}$ must vanish.
\end{proof}
Let us now investigate this success probability for $ 3 \times 3$ states. In Fig.\ \ref{3x3EsEs4Psuccfrom} the maximum success probability (Eq.\ \eqref{Psucc}) of transforming a fixed state $\ket{\phi}$ into any other $ 3 \times 3$ state $\ket{\psi}$, i.e. $P(\phi \rightarrow \psi)$ is depicted. For the contour plot we use here and in the following 0.1 steps and go from black ($P=1$) to yellow ($0 \leq P \leq 0.1$) in colors. Thus, the other colors in Fig.\ \ref{3x3EsEs4Psuccfrom} correspond to $1 < P \leq 0.9$ for gray, $0.9 < P \leq 0.8$ for blue black, $0.8 < P \leq 0.7$ for blue and $0.7 < P \leq 0.6$ for green. Note that here $P \geq 0.6$ as the state $\ket{\phi}$ reaches $\ket{\phi^+}$ with $P(\phi \rightarrow \phi^+) = 0.6$. 
\begin{figure}[h!] 
\centering
\includegraphics[width=0.4\textwidth]{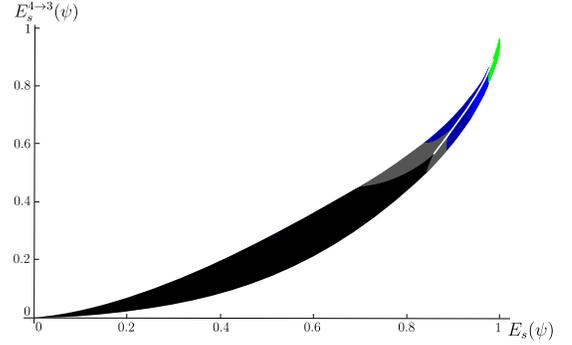}
\caption{$E_s$ versus $E_s^{4\rightarrow 3}$ including the success probability of going from a $ 3 \times 3$ state $\ket{\phi}$ with $\lambda_{\phi} = (0.52,0.28, 0.2)$ to any other state $\ket{\psi}$.}
\label{3x3EsEs4Psuccfrom}
\end{figure}
In Fig. \ref{3x3EsEs4Psuccto} we consider the opposite direction, i.e. the success probability with which any $ 3 \times 3$ state $\ket{\psi}$ can be transformed to a certain state $\ket{\phi}$, i.e. $P(\psi \rightarrow \phi)$. Note that in this case we get all possible values for the success probability ranging from $0$ to $1$. The white line in Fig.\ \ref{3x3EsEs4Psuccfrom} (\ref{3x3EsEs4Psuccto}) is parametrized by $\lambda_2 = \frac{\lambda_2^{\phi}}{\lambda_3^{\phi}} \lambda_3$ for $0 \leq \lambda_3 \leq \frac{\lambda_3^{\phi}}{2 \lambda_2^{\phi} + \lambda_3^{\phi}}$ and $\lambda_2 = \lambda_1 = \frac{1-\lambda_3}{2}$ for $\frac{\lambda_3^{\phi}}{2 \lambda_2^{\phi} + \lambda_3^{\phi}} < \lambda_3 \leq 1/3$ and indicates all states with maximum (minimum) $E_s$ and (or) $E_s^{4 \rightarrow 3}$ for constant success probability $P$ (apart from the black regions, where the success probability is equal to 1). 
\begin{figure}[h!] 
\centering
\includegraphics[width=0.4\textwidth]{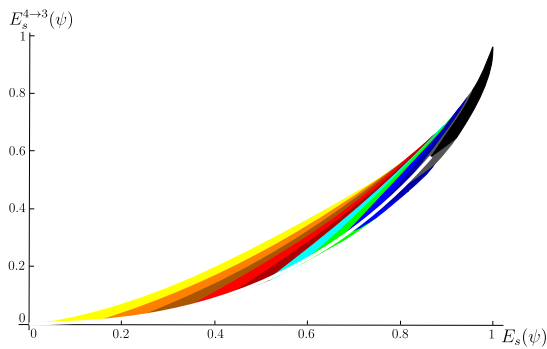}
\caption{$E_s$ versus $E_s^{4\rightarrow 3}$ including the success probability of obtaining a $ 3 \times 3$ state $\ket{\phi}$ with $\lambda_{\phi} = (0.52,0.28, 0.2)$ from any other state $\ket{\psi}$.}
\label{3x3EsEs4Psuccto}
\end{figure}
Note that one can observe the lines of constant success probability, i.e. $P = const.$,  at the boundary lines of each color transition. In Fig.\ \ref{3x3EsEs4Psuccfrom} these lines resemble the boundary lines of the set containing all states that can be reached deterministically by $\ket{\phi}$ (see Fig.\ \ref{3x3EsEs4MaMs}), which are parametrized by $E_i (\psi) = E_i(\phi)$. Thus, it is clear that the lines of constant success probability are similar to these lines, as the success probability is given by the minimal ratio of the entanglement monotones $E_i$. Hence, for a constant success probability $p_c$ we get $P(\phi \rightarrow \psi) = \min_i \frac{E_i(\phi)}{E_i(\psi)} = \frac{E_j(\phi)}{E_j(\psi)} = p_c$ and therefore $E_j(\phi) = p_c E_j(\psi)$. In Fig.\ \ref{3x3EsEs4Psuccto} the lines of constant success probability resemble the behavior of the set containing the states that can reach $\ket{\phi}$ deterministically in Fig.\ \ref{3x3EsEs4MaMs}, due to the same argument as above. \par
Next we investigate the success probability for states in terms of two measures, that do not uniquely characterize the entanglement of $ 3 \times 3$ states. We show as an example in Fig.\ \ref{3x3EsEaPsuccto} $ 3 \times 3$ states in terms of $E_s, E_a$ together with the maximum success probability of obtaining the same state $\ket{\phi}$ by any other state $\ket{\psi}$. The white line is defined as before. Moreover, in Fig.\ \ref{3x3EsEaPsuccto} the fact that $E_s$ and $E_a$ do not uniquely characterize the entanglement of $ 3 \times 3$ states can be easily noticed, as for several points with the same value of $E_s$ and $E_a$ two different success probabilities overlap. Thus, for these points there must exist two different states $\ket{\psi_1}$ and $\ket{\psi_2}$ with $E_{s(a)} ( \psi_1 ) = E_{s(a)} ( \psi_2 ) $ and $P(\psi_1) \neq P(\psi_2)$. 
\begin{figure}[h!] 
\centering
\includegraphics[width=0.4\textwidth]{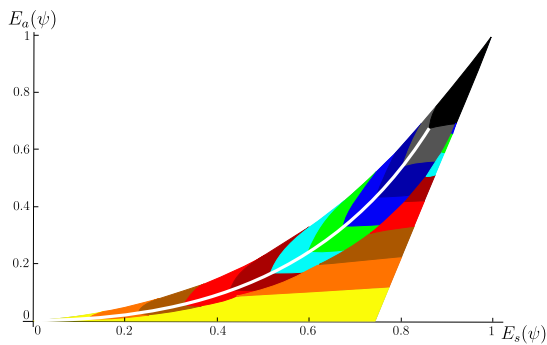}
\caption{$E_s$ versus $E_a$ including the success probability of obtaining a $ 3 \times 3$ state $\ket{\phi}$ with $\lambda_{\phi} = (0.52,0.28, 0.2)$ from any other state $\ket{\psi}$.}
\label{3x3EsEaPsuccto}
\end{figure}
Note also that we get similar results for $4 \times 4$ states. 

\section{Conclusion}
We investigated the properties and relations of two classes of operational entanglement measures, the source and accessible entanglement. We focused mainly on bipartite pure states and gave an operational characterization of bipartite pure state entanglement with the help of another operational measure, namely the geometric measure of entanglement. Furthermore, we investigated the reachable parameter regime of the source and accessible entanglement. Moreover, we determined for some fixed state $\ket{\phi}$ the parameter regime from which it can be reached and the parameter regime $\ket{\phi}$ can reach. We also showed that these regimes can be obtained analytically with the help of the Positivstellensatz and computed the boundaries of these sets. By relating the results to other entanglement measures and also to probabilistic transformations we could compare their behavior. 
\par
Let us finally mention that we investigated also $5 \times 5$ states and found that for these states comparing the source entanglement and its generalizations leads to the same results as for $ 3 \times 3$ and $4 \times 4$ states. Thus, we conclude that the simplicity of the figures showing the values of $E_s$ and the generalizations seems to be a general feature of bipartite pure states. That is the boundaries of the set of all values given by these measures can be easily obtained with the help of the states in Eq.\ \eqref{lambdaopt}. Moreover, the boundaries of the sets containing all states that can either be reached or can reach a specific state $\ket{\phi}$ are also easy to obtain and the range of possible values for the tuples $(E_s, E_s^{d+1 \rightarrow d},...,E_s^{2d -2 \rightarrow d})$ is confined. \\
It will be interesting to investigate the parameter regime of the measures considered here also for multipartite states. 

\section{Acknowledgments}
We would like to thank Gilad Gour and Richard Jozsa for helpful comments. Moreover, we thank Fedor Petrov and Tewodros Amdeberhan for helpful comments regarding the proof of Lemma \ref{LemmaEs} and in particular for pointing out Ref. \cite{Amd16} to us.
This research was funded by the Austrian Science Fund (FWF): Y535-N16.

\begin{appendix}

\section{The source and accessible entanglement of $ 3 \times 3$ and $4 \times 4$ states}
\label{sec:AppendixA}
Let us present here the explicit expressions for the source and accessible entanglement of $ 3 \times 3$ and $4 \times 4$ states, which we used in Sec. \ref{physregionEsEa}. Let us start with $ 3 \times 3$ states, for which the expressions were already obtained in \cite{Sauer15}. 
\begin{align}
&E_s(\psi) = 3\lambda_2^2-6\lambda_2 \lambda_3 - 6 (\lambda_3-1)\lambda_3,\\
&E_a(\psi) =   \begin{dcases} 12 \lambda_2 \lambda_3 & \textrm{if } \lambda_1 > \frac{1}{2} \\
12 [\lambda_2 \lambda_3 - 1/4 (1-2\lambda_1)^2] & \textrm{if } \lambda_1\leq\frac{1}{2}
  \end{dcases}
\end{align}
For the generalized source entanglement, i.e. all $4 \times 4$ states that reach a certain $ 3 \times 3$ state deterministically, we obtain
\begin{align}
E_s^{4 \rightarrow 3}( \psi) = \frac{27}{13} \left(2 \lambda_2^3+6 \lambda_2^2 \lambda_3+3 (3-4 \lambda_2) \lambda_3^2-10 \lambda_3^3\right).
\end{align}
Thus, it is moreover easy to see that $E_s$ together with $E_s^{4 \rightarrow 3}$ completely characterize the entanglement of $ 3 \times 3$ states.  \par
For $4 \times 4$ states the source and accessible entanglement are given by
\begin{align}
E_s(\psi) = & 4 \lambda_2^3+12 \lambda_2^2 \lambda_3-24 \lambda_2^2 \lambda_4-24 \lambda_2 \lambda_3^2+24\lambda_2 \lambda_3 \lambda_4 \nonumber\\
& +12 \lambda_2 \lambda_4^2-20 \lambda_3^3+12 \lambda_3^2 \lambda_4+18 \lambda_3^2 + 48 \lambda_3 \lambda_4^2 \nonumber \\
& -36 \lambda_3 \lambda_4+20 \lambda_4^3-30 \lambda_4^2+12 \lambda_4,
\end{align}
\begin{widetext}
\begin{align}
E_a(\psi) = \scriptsize{ \begin{dcases}    24 \lambda_4 (6 \lambda_2 \lambda_3+ \lambda_4 (-3 \lambda_3+
\lambda_4))  & \textrm{if } \lambda_1 \geq 1/2 \ \text{and} \
\lambda_1 > 1 - 2 \lambda_2 \\
12 \left(-\left(\lambda _2-\lambda _3\right){}^3-3 \left(\lambda _2+\lambda _3\right) \lambda _4^2+3 \lambda _4^3+3 \left(\lambda _2+\lambda _3\right){}^2 \lambda _4\right) & \textrm{if } \lambda_1 \geq \frac{1}{2} \land \lambda_1 \leq 1-2 \lambda _2 \\
 2 (-36 \lambda _1^3-18 \lambda _1 \left(\left(1-2 \lambda _2\right){}^2-2 \lambda _4^2+4 \lambda _2 \lambda _4\right)-36 \lambda _1^2 \left(\lambda _2-1\right)   & \textrm{if } \lambda_1 \leq \frac{1}{3}\land \lambda_1 \geq \frac{1}{2}-\lambda _4 \\
 +12 \left(-3 \lambda _2 \left(\lambda _4-1\right){}^2-6 \lambda _2^2 \left(\lambda _4-1\right)-4 \lambda _2^3+\lambda _4^2 \left(4 \lambda _4-3\right)\right)+5) & \\
4 (-30 \lambda _1^3+6 \left(-3 \lambda _2 \left(\lambda _4-1\right){}^2-6 \lambda _2^2 \left(\lambda _4-1\right)-4 \lambda _2^3+2 \lambda _4^3\right) & \textrm{if } \lambda_1 \leq \frac{1}{3}\land \lambda_1 \leq \frac{1}{2}-\lambda _4 \\
-18 \lambda _1 \left(\left(\lambda _4-1\right){}^2+2 \lambda _2 \left(\lambda _4-1\right)+2 \lambda _2^2\right)-9 \lambda _4-18 \lambda _1^2 \left(\lambda _2+2 \lambda _4-2\right)+4) & \\
6 (6 \lambda _1^3+12 \left(-2 \lambda _2^2+\lambda _4^2-2 \lambda _2 \left(\lambda _4-1\right)\right) \lambda _1-6 \lambda _1^2 \left(2 \lambda _2+1\right)& \textrm{if } \frac{1}{3} \leq \lambda_1 \leq \frac{1}{2}\land \lambda_1 \geq \frac{1}{2}-\lambda _4 \land\lambda_1 \leq 1-2 \lambda _2\\
+4 \left(-3 \lambda _2 \left(\lambda _4-1\right){}^2-6 \lambda _2^2 \left(\lambda _4-1\right)-4 \lambda _2^3+\lambda _4^2 \left(4 \lambda _4-3\right)\right)+1) & \\
12 \left(-\left(\lambda _1+2 \lambda _2-1\right){}^3-6 \left(\lambda _1+\lambda _2\right) \lambda _4^2+4 \lambda _4^3-3 \left(\left(1-2 \lambda _2\right){}^2+4 \lambda _1^2+4 \lambda _1 \left(\lambda _2-1\right)\right) \lambda _4\right)& \textrm{if } \frac{1}{3} \leq \lambda_1 \leq \frac{1}{2}\land \lambda_1 \leq \frac{1}{2}-\lambda _4 \land\lambda_1 \leq 1-2 \lambda _2\\
6 \left(\left(2 \lambda _1-1\right){}^3+16 \lambda _4^3+12 \lambda _4^2 \left(\lambda _1-\lambda _2-1\right)-24 \lambda _2 \left(\lambda _1+\lambda _2-1\right) \lambda _4\right)& \textrm{if } \frac{1}{3} \leq \lambda_1 \leq \frac{1}{2}\land \lambda_1 \geq \frac{1}{2}-\lambda _4 \land\lambda_1 \geq 1-2 \lambda _2\\
12 \lambda _4 \left(-12 \left(\lambda _1^2+\lambda _2^2+\lambda _1 \left(\lambda _2-1\right)\right)+4 \lambda _4^2+12 \lambda _2-6 \left(\lambda _1+\lambda _2\right) \lambda _4-3\right) & \textrm{if } \frac{1}{3} \leq \lambda_1 \leq \frac{1}{2}\land \lambda_1 \leq \frac{1}{2}-\lambda _4 \land\lambda_1 \geq 1-2 \lambda _2
 \end{dcases}}
\end{align}
\end{widetext}

Furthermore, the generalizations of the source entanglement, i.e. all $6 \times 6$ states that can reach a $4 \times 4$ state and all $5 \times 5$ states that can reach a $4 \times 4$ state are equal to
\begin{align} 
\hspace*{-0.5cm} E_s^{5 \rightarrow 4} (\psi) =& -5 (-\lambda _2^4-21 \lambda _4^4+\lambda _3^3 \left(9 \lambda _3-8\right)+4 \lambda _4^3 \left(8-15 \lambda _3\right) \nonumber\\ \nonumber &-6 \lambda _4^2 \left(\lambda _3 \left(3 \lambda _3-8\right)+2\right)+12 \lambda _3^2 \left(3 \lambda _3-2\right) \lambda _4 \\ \nonumber &-4 \lambda _2^3 \left(\lambda _3+\lambda _4\right)-6 \lambda _2^2 \left(\lambda _3^2-5 \lambda _4^2+2 \lambda _3 \lambda _4\right)\\  &+12 \lambda _2 \left(\lambda _3-\lambda _4\right) \left(\lambda _3^2+\lambda _4^2+4 \lambda _3 \lambda _4\right)), \\
\hspace*{-0.5cm} E_s^{6 \rightarrow 4} (\psi)= &6 \lambda _2^5-84 \lambda _3^5+15 \lambda _3^4 \left(5-8 \lambda _2\right)+60 \lambda _2^2 (\lambda _3^3-9 \lambda _4^3 \nonumber\\ \nonumber &+3 \lambda _4^2 \lambda _3+3 \lambda _3^2 \lambda _4)+60 \left(\lambda _2^3 (\lambda _3+\lambda _4\right){}^2+\lambda _4 (-7 \lambda _3^4 \\ \nonumber &+6 \lambda _3^2 \lambda _4^2+\lambda _3^3(-8 \lambda _2-14 \lambda _4+5)))+3 \lambda _4^3 (60 \lambda _3 \\ \nonumber &+(4 \lambda _2 +6 \lambda _4-5)+7 \lambda _4 \left(16 \lambda _4-25\right)+10 (20 \lambda _4^3\\  &+9 \lambda _3^2 \lambda _4^2 \left(5-8 \lambda _2\right)+3 \left(6 \lambda _4^4 \lambda _2+\lambda _2^4 \left(\lambda _3+\lambda _4\right)\right)).
\end{align}
Note that also for the $4 \times 4$ case we could show with the help of the Positivstellensatz that $E_s$, $E_s^{5 \rightarrow 4} $ and $E_s^{6 \rightarrow 4} $ completely characterize the entanglement.

\section{Operational characterization of bipartite entanglement}
\label{App:OpCharProof}
In this Appendix we proof Lemma\ \ref{2nx2nentanglement}, given in Sec.\ \ref{Sec:OpChar}. More precisely, we show that any bipartite pure state is operationally characterized by the bipartite entanglement of all possible bipartite splittings of some qubits in B versus the rest. Here, the bipartite entanglement is given in terms of the largest Schmidt coefficient, which is equivalent to the geometric measure of entanglement. 
 \proof{
 We show Lemma \ref{2nx2nentanglement} by induction over $n$. That is, we first consider the $n=2$ case, i.e. a $4\times 4$ state that we treat as a 4-qubit state, for which the statement in Lemma \ref{2nx2nentanglement} is easily shown as follows. The state can be written as $\ket{\psi} = \sum_{i j} \sqrt{ \lambda_{i j}} \ket{i j}_A \ket{i j}_B$. In this case we have only two different splittings (first qubit and second qubit in B versus the rest, respectively) with largest Schmidt coefficients $E_{01}= \lambda_{11}+\lambda_{01}$ and $E_{10} = \lambda_{11} + \lambda_{10}$, respectively. Thus, together with the largest Schmidt coefficient in the splitting $A$ versus $B$, i.e. $\lambda_{11}$, we obtain all Schmidt coefficients of the state $\ket{\psi}$ and therefore, uniquely characterize its entanglement. \par
Now, let us assume that the statement holds for $\ket{\Psi}\in \mathbb{C}^{2^k} \otimes \mathbb{C}^{2^k}$ for all $k \leq n-1$ and prove that we also obtain all Schmidt coefficients for $k = n$. For this we write the state as
\begin{align} \label{bipstatequbits}
\hspace*{-0.5cm}\ket{\psi} &= \sum_{i_1,i_2,...,i_n} \sqrt{ \lambda_{i_1 i_2 ... i_n}} \ket{i_1 i_2 ... i_n}_A \ket{i_1 i_2 ... i_n}_B\\
\hspace*{-0.5cm}&= \sum_{\vec{i}_{n-1}} \scriptstyle{\ket{\vec{i}_{n-1}}_{B_1...B_{n-1}}} \displaystyle{ \sum_{j_n=0}^1 \sqrt{\lambda_{\vec{i}_{n-1} j_n}}} \scriptstyle{ \ket{\vec{i}_{n-1}}_{A_1...A_{n-1}} \ket{j_{n}}_{A_n} \ket{j_{n}}_{B_n}}, \nonumber
\end{align}
with $\vec{i}_{n-1} = i_1 i_2 .... i_{n-1}$ and sorted Schmidt coefficients, i.e.  $\lambda_{0....0} \leq \lambda_{0...01} \leq \lambda_{0....10}\leq... \leq \lambda_{1...1}$.  Note that in the second line of Eq.\ \eqref{bipstatequbits} we write the state in the splitting of the first $n-1$ qubits in B versus the rest. We will consider all possible splittings of $n-1$ qubits in B, i.e. $B_{l_1,...l_{n-1}}$ versus the rest. As the second line of Eq.\ \eqref{bipstatequbits} is already the Schmidt decomposition of the state, the Schmidt coefficients are given by e.g. $E_{0...011110..0} = \lambda_{0..01110...0} + \lambda_{0...011110...0}$. At this point we use the induction assumption, namely that all Schmidt coefficients are known for the splittings $B_{l_1,...l_{n-1}}$ versus the rest, as $k = n-1$ in this case. For example for the splitting of the $n-	1$ first qubits in B versus the rest they are given by $\sum_{j_n=0}^1 \lambda_{\vec{i}_{n-1} j_n}$. The Schmidt coefficients in the splitting A versus B, i.e. the $\lambda$'s, can then be computed recursively, i.e.
\begin{align}
\lambda_{1...101...1} &= E_{1....101...1} - \lambda_{1...1}, \\
\lambda_{1...1001...1} & =E_{1....1001...1} - \lambda_{1...101...1} \nonumber \\ \nonumber &= E_{1....1001...1}-E_{1....101...1} + \lambda_{1...1},\\
\lambda_{1...10001...1} & =E_{1....10001...1} - \lambda_{1...1001...1} \nonumber \\ \nonumber &= E_{1....10001...1}-E_{1....1001...1} +E_{1....101...1}\\ \nonumber & \phantom{=} - \lambda_{1...1},\\
& \phantom{=} \vdots
\end{align}
Hence, the largest Schmidt coefficient $\lambda_{1...1}$ together with the geometric measure of entanglement in all possible splittings of some qubits in B versus the rest uniquely determine the state $\ket{\psi}$, which completes the proof.
\qed
 }
 
\end{appendix}


\begin{thebibliography}{widest-label}

\bibitem{Be} C. H. Bennett, H. J. Bernstein, S. Popescu and B. Schumacher, Phys. Rev. A \textbf{53}, 2046-2052 (1996).

\bibitem{vidal2} G. Vidal, Phys. Rev. Lett. \textbf{83}, 1046 (1999).

\bibitem{Hastings} M. Hastings, J. Stat. Mech., P08024 (2007).

\bibitem{condensed}  D. Perez-Garcia, F. Verstraete, M.M. Wolf, J.I. Cirac, 	Quantum Inf. Comput. \textbf{7}, 401 (2007); N. Laflorencie, Physics Report \textbf{643}, 1-59 (2016), and references therein.

\bibitem{Shim95} A. Shimony,  Ann. N.Y. Acad. Sci. \textbf{755}, 675 (1995).

\bibitem{Schw15} K.Schwaiger, D. Sauerwein, M. Cuquet, J.I. de Vicente, and B. Kraus, Phys. Rev. Lett. \textbf{115}, 150502 (2015).

\bibitem{Sauer15} D. Sauerwein, K. Schwaiger, M. Cuquet, J.I. de Vicente, and B. Kraus, Phys. Rev. A \textbf{92}, 062340 (2015). 

\bibitem{Parr03} P.A. Parrilo, Math. Program., Ser. B \textbf{96}, 293-320 (2003).

\bibitem{BeSa03} D. Berry, B. Sanders, J. Phys A: Math. Gen.  {\bf 36}, 12255 (2003).

\bibitem{Steng74} G. Stengle, Math. Ann. \textbf{207} 87-97 (1974).

\bibitem{Gross09} D. Gross, S.T. Flammia, J.Eisert, Phys. Rev. Lett. \textbf{102}, 190501 (2009).

\bibitem{Post09} A. Postnikov, Int. Math. Res. Not.,\textbf{ 2009}, 1026 (2009).

\bibitem{nielsen} M. A. Nielsen, Phys. Rev. Lett. \textbf{83}, 436 (1999).

\bibitem{ZyBe02} K. Zyczkowski, I. Bengtsson, Ann. Phys. {\bf 295}, 115 (2002).

\bibitem{vidal1} G. Vidal, J. Mod. Opt. \textbf{47}, 355-376 (2000).

\bibitem{horoRev} R. Horodecki \textit{et al.}, Rev. Mod. Phys. \textbf{81}, 865 (2009).

\bibitem{Wootters1} W. K. Wootters, Phys. Rev. Lett. {\bf 80}, 2245 (1998).

\bibitem{Chitambar12} E. Chitambar, W. Cui, H.-K. Lo,  Phys. Rev. A \textbf{85}, 062316 (2012)

\bibitem{Gour10} G. Gour, Phys. Rev. Lett. {\bf 105}, 190504 (2010).

\bibitem{negativity} G. Vidal, R.F. Werner, Phys. Rev. A \textbf{65}, 032314 (2002).

\bibitem{robustness} G. Vidal, R. Tarrach, Phys. Rev. A \textbf{59}, 141 (1999).

\bibitem{Coffman00} V. Coffman, J. Kundu, and W.K. Wootters, Phys. Rev. A \textbf{61}, 052306 (2000)

\bibitem{Hilbert} D. Hilbert, Math. Ann. \textbf{32}, 342 (1888). 

\bibitem{SOS} S. Prajna, A. Papachristodoulou, P.A. Parrilo, SOSTOOLS: Sum of squares optimization toolbox for MATLAB (2002). Available from: http://www.cds.caltech.edu/sostools

\bibitem{Plenio99}  D. Jonathan, M. B. Plenio, Phys. Rev. Lett. \textbf{83}, 1455 (1999).

\bibitem{Amd16} T. Amdeberhan, Proc. Amer. Math. Soc. \textbf{144}, 2799-2810 (2016).


\end{thebibliography}
\end{document}